\newcommand*{\QUANTUM}{}%
\ifdefined\QUANTUM
\documentclass[a4paper,onecolumn,11pt,accepted=2021-05-19]{quantumarticle}
\pdfoutput=1
\else
\documentclass[11pt]{article}
\fi

\usepackage[utf8]{inputenc}
\usepackage{color, xcolor, colortbl}
\usepackage{graphicx,epstopdf}
\usepackage{geometry}
\usepackage{amsmath,amssymb,amsthm}
\usepackage{algorithm}
\usepackage{algorithmic}
\usepackage{bm}
\usepackage[caption=false]{subfig}
\usepackage{appendix}
\usepackage{multirow}
\usepackage{braket}
\usepackage[english]{babel}
\usepackage{hyperref}
\usepackage[capitalize]{cleveref}
\usepackage[qm]{qcircuit}
\hypersetup{colorlinks=true, linkcolor = blue}
\usepackage{braket}
\usepackage{authblk}

\newtheorem{theorem}{Theorem}
\newtheorem{lemma}{Lemma}

\newtheorem{assump}{Assumption}
\newtheorem{rem}{Remark}

\newcommand{\REV}[1]{{#1}}

\title{Time-dependent unbounded Hamiltonian simulation with vector norm scaling}
\author[1]{Dong An}
\author[1]{Di Fang} 
\author[1,2,3]{Lin Lin}
\affil[1]{Department of Mathematics, University of California, Berkeley,  CA 94720, USA}
\affil[2]{Computational Research Division, Lawrence Berkeley National Laboratory, Berkeley, CA 94720, USA}
\affil[3]{Challenge Institute for Quantum Computation, University of California, Berkeley, CA 94720, USA}

\begin{document}

\newcommand{\diag}{\operatorname{diag}}
\renewcommand{\Re}{\operatorname{Re}}
\renewcommand{\Im}{\operatorname{Im}}
\newcommand{\conj}[1]{\overline{#1}}
\newcommand{\Tr}{\operatorname{Tr}}

\newcommand{\I}{\mathrm{i}}

\newcommand{\mc}[1]{\mathcal{#1}}
\newcommand{\mf}[1]{\mathfrak{#1}}
\newcommand{\wt}[1]{\widetilde{#1}}

\newcommand{\abs}[1]{\left\lvert#1\right\rvert}
\newcommand{\norm}[1]{\left\lVert#1\right\rVert}

\newcommand{\ud}{\,\mathrm{d}}
\newcommand{\ad}{\operatorname{ad}}

\renewcommand{\arraystretch}{1.5}

\newcommand{\Or}{\mathcal{O}}
\newcommand{\EE}{\mathbb{E}}
\newcommand{\NN}{\mathbb{N}}
\newcommand{\RR}{\mathbb{R}}
\newcommand{\CC}{\mathbb{C}}
\newcommand{\ZZ}{\mathbb{Z}}

\maketitle

\begin{abstract}
The accuracy of quantum dynamics simulation is usually measured by the error of the unitary evolution operator in the operator norm, which in turn depends on certain norm of the Hamiltonian. For unbounded operators, after suitable discretization, the norm of the Hamiltonian can be very large, which significantly increases the simulation cost. However, the operator norm measures the worst-case error of the quantum simulation, while practical simulation concerns the error with respect to a given initial vector at hand. We demonstrate that under suitable assumptions of the Hamiltonian and the initial vector, if the error is measured in terms of the vector norm, the computational cost may not increase at all as the norm of the Hamiltonian increases using Trotter type methods.  In this sense, our result outperforms all previous error bounds in the quantum simulation literature. Our result extends that of [Jahnke, Lubich, BIT Numer. Math. 2000] to the time-dependent setting. We also clarify the existence and the importance of commutator scalings of Trotter and generalized Trotter methods for time-dependent Hamiltonian simulations.
\end{abstract}

\tableofcontents 

\section{Introduction}

Simulation of the quantum dynamics is widely viewed as one of the most important applications of a quantum computer. Let $H(t)$ be a Hamiltonian defined on the interval $[0,T]$, and $\ket{\psi_0}$ be the initial vector, then the time-dependent Hamiltonian simulation problem aims to find $\ket{\psi(T)}$, which solves the time-dependent Schr\"odinger equation
\begin{equation}
    \I\partial_t \ket{\psi(t)} = H(t)\ket{\psi(t)}, \quad \ket{\psi(0)} = \ket{\psi_0}.
     \label{eqn:td_hamsim}
\end{equation}
If $H(t)\equiv H$ is time-independent, then the solution can be expressed in closed form as $\ket{\psi(T)}=\exp(-\I TH)\ket{\psi_0}$. The task of creating such a universal quantum simulator was first conceptualized by Lloyd \cite{Lloyd1996}, and the past few years have witnessed significant progresses in the development of new quantum algorithms as well as the improvement of theoretical error bounds of existing quantum algorithms for time-independent Hamiltonian simulation \cite{BerryAhokasCleveEtAl2007,BerryChilds2012,BerryCleveGharibian2014,BerryChildsCleveEtAl2015,BerryChildsKothari2015,LowChuang2017,ChildsMaslovNamEtAl2018,LowWiebe2019,ChildsOstranderSu2019,Campbell2019,Low2019,ChildsSu2019,ChildsSuTranEtAl2020,ChenHuangKuengEtAl2020,SahinogluSomma2020}. In particular, for a $d$-sparse Hamiltonian with bounded $\norm{H}_{\max}$ (the largest element of $H$ in absolute value), the complexity of the quantum signal processing (QSP) method \cite{LowChuang2017} is $\Or\left(T d\norm{H}_{\max }+\log (1 / \epsilon) / \log \log (1 / \epsilon)\right)$, which matches complexity lower bounds in all parameters. Meanwhile the error bound of the high order Trotter-Suzuki scheme has also been significantly improved \cite{ChildsSuTranEtAl2020}, which yields near-best asymptotic complexities for simulating problems such as $k$-local Hamiltonians, and improves previous error bounds for Hamiltonians with long range interactions.

On the other hand, simulation with time-dependent Hamiltonians appears ubiquitously, such as in the context of quantum controls \cite{ZhuRabitz1998,MadaySalomonTurinici2006,NielsenDowlingGuEtAl2006,DongPetersen2010,PangJordan2017,MagannGraceRabitzEtAl2020}, non-adiabatic quantum dynamics \cite{RungeGross1984,CurchodMartinez2018}, and adiabatic quantum computation \cite{FarhiGoldstoneGutmannEtAl2000,RolandCerf2002,AlbashLidar2018}, to name a few. Compared to the time-independent setting, there are significantly fewer quantum algorithms available \cite{HuyghebaertDeRaedt1990,WiebeBerryHoyerEtAl2010,PoulinQarrySommaEtAl2011,BerryChildsCleveEtAl2014,BerryChildsCleveEtAl2015,WeckerHastingsWiebeEtAl2015,LowWiebe2019,BerryChildsSuEtAl2020}. The time-dependent setting has so far ruled out the usage of quantum signal processing \cite{LowChuang2017} and quantum singular value transformation (QSVT) \cite{GilyenSuLowEtAl2019} types of techniques. To the extent of our knowledge, the best results available are given by the truncated Dyson series simulation \cite{BerryChildsCleveEtAl2015,LowWiebe2019}, and the rescaled Dyson series method \cite{BerryChildsSuEtAl2020}, which scales with respect to $\|H\|_{\max,\infty} := \sup_{t\in[0,T]} \|H(t)\|_{\max}$, and $\|H\|_{\max,1}:=\int_{0}^{T} \norm{H(t)}_{\max} dt$, respectively.

In this paper, we are concerned with the simulation of a time-dependent and unbounded Hamiltonian $H(t)$, which naturally includes the simulation of a time-independent Hamiltonian $H(t)\equiv H$ as a special case. More precisely, we assume that there is a family of Hamiltonians $H^{(n)}(t)$ such that as $n\to \infty$, the norm of $H$ (e.g. the max of the operator norm or the $L^1$ norm) also increases towards infinity. 

For concreteness, we will consider the bilinear quantum control Hamiltonian of the following form
\begin{equation}
    H^{(n)}(t) = f_1(t)H^{(n)}_1 + f_2(t)H^{(n)}_2.
    \label{eqn:control_ham}
\end{equation}
Here $H^{(n)}_1$ and $H^{(n)}_2$ are time-independent Hamiltonians, and $f_1$ and $f_2$ are two bounded, smooth scalar functions on a time interval $[0,T]$. Without loss of generality we assume that $\lim_{n\to \infty}\norm{H_1^{(n)}}=\infty$, while $\lim_{n\to \infty}\norm{H_2^{(n)}}<\infty$, i.e. $H_1^{(n)}$ approaches an unbounded operator, while the limit of $H_2^{(n)}$ is a bounded operator. We also assume that $\exp\left(-\I H_1^{(n)}\right)$ and $\exp\left(-\I H_2^{(n)}\right)$ can be efficiently simulated. More specifically, if $n$ also denotes the dimension of the Hamiltonian, we assume that the cost of the time-independent simulations depends at most poly-logarithmically in terms of $n$ and the error $\epsilon$.  Such an assumption is standard for $H_2^{(n)}$ because $H_2^{(n)}$ has spectral norm asymptotically independent of $n$ thus can be efficiently simulated, e.g. via the QSP technique~\cite{LowChuang2017}. However, the assumption on the effectiveness of simulating $H_1^{(n)}$ is very strong especially when $\|H_1^{(n)}\|$ grows polynomially in terms of $n$, and the no-fast-forwarding theorem~\cite{BerryAhokasCleveEtAl2007,BerryChildsKothari2015} requires roughly $\Omega(\|H_1^{(n)}\|)$ queries for generic quantum algorithms to simulate $\exp\left(-\I H_1^{(n)}\right)$. Nevertheless, for a subset of Hamiltonians with special structures, such a \textit{time-independent} simulation can indeed be fast-forwarded and the query complexity is still poly-logarithmic of $n$. Typical examples include 1-sparse Hamiltonians~\cite{childs2003exponential,ahokas2004improved,LowWiebe2019} and thus unitarily diagonalized Hamiltonians where the diagonalization procedure can be efficiently implemented. We will show later the $H_1^{(n)}$ of interest in this paper can also be fast-forwarded. The availability of the fast-forwarded time-independent Hamiltonian simulation allows us to measure the cost directly in terms of the number of Trotter steps. 

When the context is clear, we will drop the superscript $n$ and assume instead that $\norm{H_1}$ is sufficiently large. In particular, we have $\norm{H_1}\gg \norm{H_2}$.  The form of \cref{eqn:control_ham} allows us to efficiently evaluate terms of the form $\int_{t_1}^{t_2} H(t) dt = \left(\int_{t_1}^{t_2} f_1(t) dt\right) H_1 + \left(\int_{t_1}^{t_2} f_2(t) dt\right) H_2$, where the coefficients in the parentheses can be precomputed on classical computers when $f_1(t),f_2(t)$ are available.  

As an example, consider the following Schr\"odinger equation with a time-dependent effective mass $M_{\text{eff}}(t)$ (see e.g. \cite{DantasPedrosaEtAl1992,PedrosaGuedes1997,Pedrosa1997,JiKimEtAl1995,Feng2001,SchulzeHalberg2005}) in a domain $D$ with proper boundary conditions as
\begin{equation}\label{eqn:schrodinger_tdmass}
    H(t) = -\frac{1}{2M_{\text{eff}}(t)} \Delta + \frac{1}{2}M_{\text{eff}}(t) \omega^2(t) V(x), \quad x\in D.
\end{equation}
Here $\omega(t)$ is a frequency parameter.
Then we set $f_1(t)=1/(2M_{\text{eff}}(t)),f_2(t)=M_{\text{eff}}(t)\omega^2(t)/2$. When $V(x)\equiv x^2$ the system is a quantum harmonic oscillator with time-dependent effective mass. In general we assume the potential has suitable regularity conditions and is bounded on $D$. After proper spatial discretization using $n$ degrees of freedom, $H_1^{(n)}$ is the discretized negative Laplacian operator  $-\Delta$ which is unbounded, and $H_2^{(n)}$ is the discretized diagonal potential $V(x)$ which is bounded. 
We notice that the simulation of $H_1^{(n)}$ can be fast-forwarded since it can be diagonalized under the quantum Fourier transform procedure~\cite{NielsenChuang2000}.
In order to demonstrate the behavior of the Trotter formulae for unbounded operators, we require $n$ to grow polynomially with respect to $\epsilon^{-1}$, where $\epsilon$ is the relative 2-norm error of the solution. This is the case, for instance, when the potential $V(x)$ is of limited regularity. Throughout the paper we only require $V(x)$ to be a $C^4$ function on the domain $D$.\footnote{Here the $C^4$ regularity is a technical assumption to bound the norm of the nested commutators, which will be detailed in \cref{sec:harmonic}.} Again for concreteness of discussion about the computational cost, unless otherwise specified, we will assume the system is one-dimensional, $D=[0,1]$ with the periodic boundary condition, and use the second order finite difference method with $n$ equidistant nodes for spatial discretization.

To the extent of our knowledge, all previous results in the quantum simulation literature (for both time-independent and time-dependent Hamiltonians) measure the error of the evolution operator $\norm{\wt{U}(T)-U(T)}$, where $U(T)=\exp_{\mc{T}}\left(-\I\int_{0}^T H(t) dt\right)$ is the exact evolution operator expressed in terms of a time-ordered matrix exponential, and $\wt{U}(T)$ is an approximate evolution operator obtained via the numerical scheme. We then directly obtain the vector norm error $\norm{\ket{\wt{\psi}(T)}-\ket{\psi(T)}}\le \norm{\wt{U}(T)-U(T)}$. 
However, since $\norm{\wt{U}(T)-U(T)}$ typically depends polynomially on the operator norm $\|H_1\|$, as $\norm{H_1}$ increases, if the computational cost does not increase accordingly, then all error bounds of the operator norm $\norm{\wt{U}(T)-U(T)}$ would increase to $\Or(1)$, 
with the exception of the interaction picture method for time-independent Hamiltonian simulations~\cite{LowWiebe2019}.\footnote{In the context of time-independent simulation~\cite{ChildsSuTranEtAl2020}, the error of Trotter methods does not scale directly with respect to the operator norm $\norm{H_1}$, but with respect to the norm of the (high-order) commutators  $[H_1,H_2],[H_1,[H_1,H_2]],[H_2,[H_2,H_1]]$ and so on. However, this does not change our conclusion here. In principle, the interaction picture method can also be generalized to efficiently simulate time-dependent Hamiltonians. However, its practical performance has not been well understood. } 
While the operator norm error provides an upper bound of the error given \textit{any} initial vector, for a \textit{particular} simulation instance, it is the vector norm error $\norm{\ket{\wt{\psi}(T)}-\ket{\psi(T)}}$ that matters. It turns out that for certain unbounded operators and initial vectors, the vector norm bound can be significantly improved. The key reason is that the magnitude of terms such as $\norm{H_1 \ket{\psi}},\norm{[H_1,H_2]\ket{\psi}}$ can be much smaller than the corresponding operator norm estimates.
In fact the importance of the vector norm estimates has long been recognized in the numerical analysis literature, and the vector norm error bounds have been established for time-independent Hamiltonian simulation using second and higher order Trotter methods of the form $H=-\Delta+V(x)$~\cite{JahnkeLubich2000,Thalhammer2008,DescombesThalhammer2010}, and for time-dependent Hamiltonian simulation using Magnus integrators of the form $H = -\Delta + V(t,x)$~\cite{HochbruckLubich2003}.  
Under suitable discretization and choice of the initial vector, the vector norm error $\norm{\ket{\wt{\psi}(T)}-\ket{\psi(T)}}$ obtained by the standard Trotter method remains small, even as $\norm{H}\to \infty$ and the operator norm $\norm{\wt{U}(T)-U(T)}$ becomes $\Or(1)$.

\paragraph{Contribution:} 

The \textit{first contribution} of this work is to extend the vector norm estimate \cite{JahnkeLubich2000} to time-dependent unbounded Hamiltonian simulations.  For concreteness we focus on the standard first and second-order Trotter methods, as well as a class of generalized Trotter methods proposed in \cite{HuyghebaertDeRaedt1990}, which will be introduced in \cref{sec:trotter_error_rep}. Our main result for a given control Hamiltonian \cref{eqn:control_ham} is \cref{thm:vector_norm_bound}. It states that under suitable assumptions, the vector norm error obtained from both standard and generalized Trotter methods depends mainly on $\sup_{t\in[0,T]}\|H_1\ket{\psi(t)}\|$, which can be significantly smaller than $\norm{H_1}$.\footnote{\cref{thm:vector_norm_bound} shows that the number of Trotter steps may not scale with respect to $\|H_1\|$. The mechanism of the improvement is very different from that of the interaction picture approach, where the number of the time steps is still linear in $\|H_1\|$. }

In order to simulate the Hamiltonian of the form \cref{eqn:schrodinger_tdmass}, we take both the spatial and temporal discretization into account, and our complexity estimates are given in \cref{thm:main_harmonic}. Our result compared to existing results are given in \cref{tab:trotter_comparison}, where the complexity for time-independent simulations are obtained by treating $M_{\text{eff}}(t),\omega(t)$ as constants. In particular, the vector norm is asymptotically independent of the spatial discretization parameter $n$, and complexity in terms of the error matches that of the time-independent Hamiltonian simulation obtained by \cite{JahnkeLubich2000}. 
Under the same second order spatial discretization, our complexity estimate for second order Trotter formulae outperforms state-of-the-art error bounds using high order Trotter and post-Trotter schemes~\cite{BerryChildsCleveEtAl2015,LowWiebe2019,BerryChildsSuEtAl2020} in terms of the desired level of accuracy, due to their dependence on the spectral norm of $H_1$ and thus on $n$.


The effectiveness of the vector norm bound depends on the initial vector $\ket{\psi_0}$. We remark that recently \cite{SahinogluSomma2020} establishes improved error estimates of low-order Trotter methods for time-independent Hamiltonian simulation, when the initial vector is constrained to be within a low energy subspace. Another recent work~\cite{SuHuangCampbell2020} obtains an improved complexity estimate for simulating a system with $\eta$ interacting electrons using time-independent Trotter formula, by considering the operator norm constrained on this $\eta$-electron sub-manifold. Our vector norm estimate provides a complementary perspective in understanding why such improved estimates are possible. When $\sup_{t\in[0,T]}\|H_1\ket{\psi(t)}\|$ is indeed comparable to $\norm{H_1}$, the operator norm bound still serves as a good indicator of the error. 

Given the improved error commutator scaling estimates for time-independent simulations~\cite{ChildsSuTranEtAl2020}, it is natural to ask whether the commutator scaling of the operator norm still holds for time-dependent simulations. The \textit{second contribution} of this paper is to reveal that for time-dependent simulations, the error of standard Trotter method \textit{does not} exhibit commutator scalings, while the commutator scaling holds for the generalized Trotter method (\cref{thm:operator_norm_bound}). 
Therefore in the context of time-dependent simulations, the use of the generalized Trotter method could reduce the simulation cost. 
Our proof of the operator norm error bounds mainly follow the procedure proposed in~\cite{ChildsSuTranEtAl2020}, and our results generalize the first and second order time-independent results in~\cite{ChildsSuTranEtAl2020} in the sense that, when the scalar functions $f_1$ and $f_2$ are constant functions, both time-dependent standard Trotter formula and time-dependent generalized Trotter formula degenerate to the same time-independent Trotter formula, and the corresponding operator norm error bound is of commutator scaling.

Yet another twist comes when we ask the question: when $H_1$ is unbounded, is it clear that the norm of the commutators $\norm{[H_1,H_2]},\norm{[H_1,[H_1,H_2]]},\norm{[H_2,[H_2,H_1]]}$ must be smaller than $\norm{H_1}$? It turns out that for the Hamiltonian \cref{eqn:schrodinger_tdmass}, we may directly analyze that 
$\norm{[H_1,H_2]},\norm{[H_2,[H_2,H_1]]}\in \Or(n)$, while $\norm{[H_1,[H_1,H_2]]},\norm{H_1}\in \Or(n^2)$ (see \cref{sec:harmonic}). Therefore the first-order generalized Trotter method outperforms the first-order standard Trotter method, but the asymptotic efficiency of the second-order generalized and standard Trotter methods are the same (\cref{lem:harmonic_time_error}). \cref{tab:trotter_summary} summarizes the performance of Trotter and generalized Trotter methods. Though both second-order schemes share the same asymptotic scaling, the generalized Trotter formula may still be a better choice in practice due to smaller preconstants, which is observed numerically. Moreover, the $p$-th generalized Trotter scheme depends only on the $(p-1)$-th derivatives of the control functions, while the $p$-th standard Trotter method on its $p$-th derivative. Therefore when the control functions have high frequency or limited regularity, the generalized Trotter scheme may significantly outperform the standard one. Such an advantage under first-order schemes has been demonstrated in~\cite{HuyghebaertDeRaedt1990} as well. 

All results above are confirmed by numerical experiments for the model \cref{eqn:schrodinger_tdmass} in \cref{sec:numer}, which verifies the sharpness of our  estimates. 

\begin{table}[]
    \centering
    \begin{tabular}{p{3cm}|p{5.5cm}|p{2.4cm}|p{2.4cm}}\hline\hline
         & Work/Method & Scaling w.  spatial discretization & Overall query complexity \\\hline
       \multirow{2}{8em}{Time-independent second order Trotter} & Childs \textit{et al}.~\cite{ChildsSuTranEtAl2020}  & $\Or(n)$ & $\Or(\epsilon^{-1})$ \\
       \cline{2-4} & Jahnke \textit{et al}.~\cite{JahnkeLubich2000} & $\Or(1)$ & $\Or(\epsilon^{-0.5})$ \\\hline
       \multirow{4}{8em}{Time-independent higher order methods} & $p$-th order Trotter~\cite{ChildsSuTranEtAl2020} & $\Or(n^{2-2/p})$ & $\Or(\epsilon^{-1})$ \\
       \cline{2-4} & Truncated Taylor series~\cite{BerryChildsCleveEtAl2014,KivlichanWiebeBabbushEtAl2017} & $\widetilde{\Or}(n^2)$ & $\widetilde{\Or}(\epsilon^{-1})$ \\
       \cline{2-4} & Quantum signal processing~\cite{LowChuang2017} & ${\Or}(n^2)$ & ${\Or}(\epsilon^{-1})$ \\
       \cline{2-4} & Interaction picture~\cite{LowWiebe2019} & $\Or(\log(n))$ & $\Or(\operatorname{polylog}(1/\epsilon))$ \\\hline
       \multirow{3}{8em}{Time-dependent second order Trotter} & Huyghebaert \textit{et al}.~\cite{HuyghebaertDeRaedt1990} & $\Or(n)$ & $\Or(\epsilon^{-1})$ \\
       \cline{2-4} & Wiebe \textit{et al}.~\cite{WiebeBerryHoyerEtAl2010} & $\Or(n^3)$ & $\Or(\epsilon^{-2})$  \\
       \cline{2-4} & Wecker \textit{et al}.~\cite{WeckerHastingsWiebeEtAl2015} & $\Or(n)$ & $\Or(\epsilon^{-1})$  \\
       \cline{2-4} & This work & $\Or(1)$ & $\Or(\epsilon^{-0.5})$ \\\hline
       \multirow{2}{8em}{Time-dependent higher order methods} & $p$-th order Trotter~\cite{WiebeBerryHoyerEtAl2010} & $\Or(n^{2+2/p})$ & $\Or(\epsilon^{-1-2/p})$ \\
       \cline{2-4} & Truncated Dyson series~\cite{BerryChildsCleveEtAl2015,LowWiebe2019} & $\widetilde{\Or}(n^2)$ & $\widetilde{\Or}(\epsilon^{-1})$ \\
       \cline{2-4} & Rescaled Dyson series~\cite{BerryChildsSuEtAl2020} & $\widetilde{\Or}(n^2)$ & $\widetilde{\Or}(\epsilon^{-1})$ \\\hline\hline
    \end{tabular}
    \caption{Comparison of complexity estimates for simulating the model  \cref{eqn:schrodinger_tdmass} in one-dimension using second order Trotter method or higher order Trotter or post-Trotter method, with $C^4$ potential function $V(x)$, and time-independent mass and frequency (top 2) or time-dependent mass and frequency (bottom 2). 
    For all the methods, we use a second order finite difference discretization with $n$ degrees of freedom. 
    The third column summarizes the scaling of the cost with respect to $n$ in order to reach constant target accuracy, and the fourth column summarizes the overall query complexity in order to achieve a desired level of relative 2-norm error $\epsilon$. Since we assume the efficiency of time-independent simulation for both $H_1$ and $H_2$, the query complexity is measured by the number of required Trotter steps for Trotter-type methods, or the query complexity under standard query model for post-Trotter methods. The simulation time $T$ is $\Or(1)$. `This work' refers to the vector norm error bound using the second order standard or generalized Trotter formula. Throughout the paper $f=\wt{\Or}(g)$ if $f=\Or(g\operatorname{polylog}(g))$. See \cref{app:derivation_table} for details of the derivation of the scalings. }
    \label{tab:trotter_comparison}
\end{table}

\begin{table}[]
    \centering
    \begin{tabular}{p{2cm}|p{4.5cm}|p{2.7cm}|p{2.7cm}}\hline\hline
         & Method \& Error type & Scaling w. spatial discretization & Overall number of Trotter steps\\\hline
       \multirow{3}{8em}{First-order Trotter} & standard, operator norm & $\Or(n^2)$ & $\Or(\epsilon^{-2})$ \\
       \cline{2-4} & generalized, operator norm & $\Or(n)$ & $\Or(\epsilon^{-1.5})$ \\
       \cline{2-4} & standard, vector norm & $\Or(1)$ & $\Or(\epsilon^{-1})$ \\
       \cline{2-4} & generalized, vector norm & $\Or(1)$ & $\Or(\epsilon^{-1})$ \\\hline
       \multirow{3}{8em}{Second-order Trotter} & standard, operator norm & $\Or(n)$ & $\Or(\epsilon^{-1})$ \\
       \cline{2-4} & generalized, operator norm & $\Or(n)$ & $\Or(\epsilon^{-1})$ \\
       \cline{2-4} & standard, vector norm & $\Or(1)$ & $\Or(\epsilon^{-0.5})$ \\
       \cline{2-4} & generalized, vector norm & $\Or(1)$ & $\Or(\epsilon^{-0.5})$ \\\hline\hline
    \end{tabular}
    \caption{Summary of results for first and second order Trotter formulae applied to simulating the model \cref{eqn:schrodinger_tdmass} in one-dimension with time-dependent effective mass and frequency. For all the methods, we use a second order finite difference discretization with $n$ degrees of freedom.  
    The third column summarizes the scaling of the cost with respect to $n$ in order to reach constant target accuracy (\cref{lem:harmonic_time_error}), and the fourth column summarizes the overall number of required Trotter steps to achieve a desired level of relative 2-norm error $\epsilon$ estimated from error bounds in different norms (\cref{thm:main_harmonic}). The simulation time $T$ is $\Or(1)$. }
    \label{tab:trotter_summary}
\end{table}

\section{Preliminaries}

In this section we first introduce several notations and preliminary lemmas used in this paper. 
Then we briefly sketch the main ideas for proving the main theorems. 

\subsection{Notations} 

We refer to a (possibly unnormalized) vector as $\vec{\psi}, \vec{u}$ or $\vec{v}$ depending on the context, and use $\ket{\psi}$ to denote the corresponding quantum state (\emph{i.e.} normalized vector under vector 2-norm). 
\REV{We define two vector norms for a vector $\vec{\psi} = (\psi_0,\cdots,\psi_{n-1})$, namely the standard 2-norm
\begin{equation*}
    \|\vec{\psi}\| = \sqrt{\sum_{k=0}^{n-1}|\psi_k|^2}, 
\end{equation*}
and the rescaled 2-norm 
\begin{equation*}
    \|\vec{\psi}\|_{\star} = \frac{1}{\sqrt{n}}\|\vec{\psi}\|. 
\end{equation*} 
The rescaled 2-norm is directly motivated by the discretization of the continuous $L^2$ norm~\cite{LeVeque2007, Thomas1995}. 
Specifically, for a real-space function $u(x)$ discretized in the real space using $n$ equidistant nodes, we apply the trapezoidal rule and obtain
\begin{equation}
        \int_0^1 |u(x)|^2 dx  \approx \sum_{k=0}^{n-1} \left(|u(k/n)|^2\frac{1}{n}\right)= \frac{1}{n}\|(u(k/n))_{k=0}^{n-1}\|^2 = \|(u(k/n))_{k=0}^{n-1}\|_{\star}^2.
\end{equation}}
Since the $\|\cdot\|_{\star}$ simply rescales the \REV{standard} vector 2-norm, the estimates that we will derive for 2-norm also hold for this rescaled 2-norm. 
Furthermore, the corresponding matrix norm induced by the rescaled 2-norm is still the \REV{standard} matrix 2-norm without any rescaling factor, as 
$$\|A\| = \sup_{\|\vec{u}\| \neq 0} \|A\vec{u}\|/\|\vec{u}\| = \sup_{\|\vec{u}\|_{\star} \neq 0} \|A\vec{u}\|_{\star}/\|\vec{u}\|_{\star}.$$ 
We remark that it is equivalent to use either 2-norm or rescaled 2-norm if we wish to bound the relative error of the numerical solutions. 

For two matrices $A,B$, define the adjoint mapping $\ad_A$  as
\begin{equation}
    \ad_A(B) = [A,B] = AB-BA,
\end{equation}
and then the conjugation of matrix exponentials of the form $\exp(A)B\exp(-A)$ can be simply expressed as  
\begin{equation}
    \exp(\ad_A)B = \exp(A)B\exp(-A).
\end{equation}
The following conjugation of matrix exponentials will be commonly used for a scale-valued function $f$ and matrices $A,B$
\begin{equation}
    \exp\left(\ad_{\I\int_{t_1}^{t_2} f(s)ds A}\right)B = \exp\left(\I\int_{t_1}^{t_2} f(s)ds A\right)B\exp\left(-\I\int_{t_1}^{t_2} f(s)ds A\right). 
\end{equation}
\REV{For a scalar-valued continuous function $f(t)$ defined on time domain $t\in[0,T]$, we use $\|f\|_{\infty}$ to denote the supremum of the function in this time interval, \emph{i.e.} 
$$\|f\|_{\infty} = \sup_{t\in[0,T]}|f(t)|.$$ }

\subsection{Elementary lemmas}

We review two elementary lemmas to be used in the proof of the paper. 
Proofs of the results can be found in, \emph{e.g. }~\cite{HairerNorsettWanner1987,Knapp2005}. 

\begin{lemma}[Taylor's theorem]\label{lem:Taylor}
    For any $k$-th order continuously differentiable function $f$ (scale-valued or matrix-valued) defined on an interval $[a,t]$, we have 
    \begin{equation}
        f(t) = \sum_{j=0}^{k-1} \frac{f^{(j)}(a)}{j!}(t-a)^j + \int_a^t \frac{f^{(k)}(s)}{(k-1)!} (t-s)^{k-1}ds. 
    \end{equation}
\end{lemma}

\begin{lemma}[Variation of parameters formula]\label{lem:VoP}
    Assume $U(t,s)$ solves the differential equation    
    \begin{equation}
        \partial_t U(t,s) =  H(t) U(t,s), \quad U(s, s) = I.
    \end{equation}
    Then 
    \begin{enumerate}
        \item For any matrix-valued continuous function $R(t)$, the solution of the differential equation
    \begin{equation}
        \partial_t \widetilde{U}(t,0) =  H(t) \widetilde{U}(t,0) + R(t), \quad \widetilde{U}(0, 0) = I
    \end{equation}
    can be represented as 
    \begin{equation}
        \widetilde{U}(t,0) = U(t,0) + \int_0^t U(t,s)R(s) ds. 
    \end{equation}
    \item For any vector-valued continuous function $\vec{r}(t)$, the solution of the differential equation
    \begin{equation}
        \partial_t \vec{\widetilde{u}}(t) =  H(t) \vec{\widetilde{u}}(t) + \vec{r}(t), \quad \vec{\widetilde{u}}(0) = \vec{u}_0
    \end{equation}
    can be represented as 
    \begin{equation}
        \vec{\widetilde{u}}(t) = U(t,0)\vec{u}_0 + \int_0^t U(t,s)\vec{r}(s)ds. 
    \end{equation}
    \end{enumerate}
\end{lemma}

\subsection{Main ideas}\label{sec:main_idea}

Here we discuss the main ideas for our operator and vector norm error bounds, and they are applicable to both standard and generalized Trotter formulae to be introduced in \cref{sec:trotter_methods}. 
\REV{For simplicity, we only discuss first order formulae here, and the ideas for second order formulae are similar. }
First, in \cref{sec:trotter_error_rep}, we derive the error representations between the exact evolution operator $U(h,0)$ and the Trotterized evolution operator $U_{s}(h,0)$ or $U_{g}(h,0)$, by establishing the differential equations that these unitary operators satisfy and by using variation of parameters. 
Such error representations are exact. 
Furthermore, although full error representations can be technically complicated, they are simply linear combinations of integrals with integrand of the form 
\begin{equation}\label{eqn:trotter_error_rep_form}
    \left(\prod_{j=1}^J g_j \right)\left[\prod_{k=1}^{K}\exp\left(\I h \xi_k H_{l_k}\right) \right]  A  \left[\prod_{k'=1}^{K'}\exp\left(\I h \xi'_{k'} H_{l'_{k'}}\right) \right].
\end{equation}
Here functions $g_j$'s can be $f_1$, $f_2$, or their derivatives; 
$H_{l_k}$ and $H_{l'_{k'}}$ are either $H_1$ or $H_2$; and 
$\xi_k$ and $\xi'_{k'}$ are some bounded real numbers. The matrix
$A$ is in the set 
\REV{$\{H_1,H_2,[H_1,H_2]\}$} for standard Trotter formula and 
\REV{can only be $[H_1, H_2]$} for generalized Trotter formula. 
Therefore it suffices to focus on each term in the form of \cref{eqn:trotter_error_rep_form} to obtain error bounds. 

The operator norm error bounds directly follow the error representations. 
Under the assumption that $H_1$ and $H_2$ are bounded operators, we can simply bound all the unitaries by $1$ and the \REV{(local)} operator norm error bounds become $\alpha h^2$ where $\alpha$ can be expressed in terms of \REV{$\|H_1\|,\|H_2\|$ and $\|[H_1,H_2]\|$} for the standard Trotter formula, and of \REV{$\|[H_1,H_2]\|$} for the generalized Trotter formula, respectively.
Notice that the local errors for both formulae are \REV{$\Or(h^2)$} for \REV{first} order schemes, which agrees with the order condition. 
Furthermore, the preconstant for the generalized Trotter formula only consists of commutators, while the preconstant for standard Trotter formula still includes norms of $H_1$ and $H_2$ themselves. 

\REV{Next we focus on the situation when $\|H_1\|$ is very large, and we still would like to obtain a well approximated quantum state of the exact wavefunction. 
In this case, the operator norm error bounds do not offer useful performance guarantees. 
To obtain a vector norm error bound, the starting point of our approach is still the exact error representation. 
Notice that the error between the exact state $\ket{\psi}$ and the approximate state $\ket{\wt{\psi}}$ obtained by Trotter formulae can be expressed as $\norm{\ket{\wt{\psi}(h)}-\ket{\psi(h)}}\le \norm{\left(\wt{U}(h,0)-U(h,0)\right)\ket{\psi(0)}}$. 
Therefore the vector norm error bounds should be a linear combination of the terms of the form 
\begin{equation}\label{eqn:trotter_error_rep_form_vec}
 \left\|\left[\prod_{k=1}^{K}\exp\left(\I h \xi_k H_{l_k}\right) \right]  A  \left[\prod_{k'=1}^{K'}\exp\left(\I h \xi'_{k'} H_{l'_{k'}}\right) \right]\vec{v}\right\| = \left\|A  \left[\prod_{k'=1}^{K'}\exp\left(\I h \xi'_{k'} H_{l'_{k'}}\right) \right]\vec{v}\right\|. 
\end{equation}
This can be obtained by applying the operator \cref{eqn:trotter_error_rep_form} to some vector $\vec{v}$, and $\vec{v}$ is related to the initial condition as well as the exact Schr\"odinger wavefunctions. }

\REV{To further bound \cref{eqn:trotter_error_rep_form_vec}, the key observation is as follows. When $A$ is $H_1$ (or commutators involving $H_1$), although $\norm{A}$ can be very large, it is possible for $\norm{A\vec{v}}$ to be small for certain vectors $\vec{v}$. 
As an example, let us consider the continuous case in one dimension and we take $H_1 = -\Delta$ and $H_2$ being a bounded, smooth potential function $V(x)$. 
The direct computation shows that 
$$H_1\psi = -\partial_x^2 \psi, $$
$$[H_1,H_2]\psi = [-\Delta, V]\psi = - (\partial_x^2 V)\psi - 2(\partial_x V)\partial_x \psi. $$
Both of the terms on the right hand side depend on the spatial derivatives of the wavefunctions, of which the norm can be small if $\psi$ is a smooth function. 
Then according to \cref{eqn:trotter_error_rep_form_vec}, we only need to somehow exchange the order between $A$ and the exponentials without introducing much overhead (\cref{lem:vector_norm}). Combining all previous arguments, we can obtain the desired vector norm error bounds.}

\section{Trotter type algorithms and error representations}\label{sec:trotter_error_rep}


In this section, we consider two different types of Trotter algorithms -- the standard and generalized Trotter formulae -- in simulating time dependent Hamiltonian \cref{eqn:control_ham}, and derive their error representations explicitly. 
Here for simplicity we restrict ourselves to the first-order and second-order cases. 
We point out that such schemes and results regarding the (non-)existence of commutator scaling can be generalized to their higher order counterparts. 

\subsection{The standard and generalized Trotter formulae}\label{sec:trotter_methods}

Both the standard and the generalized Trotter formulae (proposed in~\cite{HuyghebaertDeRaedt1990}) belong to the class of splitting methods \cite{HairerLubichWanner2006}.

The first-order standard Trotter algorithm is
\begin{equation} \label{eq:scheme_s_trotter_1st}
    U_{s,1}(t+h,t) = \exp\left(-\I f_2(t+h)H_2h\right)\exp\left(-\I f_1(t+h)H_1h\right). 
\end{equation}
The first-order generalized Trotter formula  is
\begin{equation}\label{eq:scheme_g_trotter_1st}
    U_{g,1}(t+h,t) = \exp\left(-\I\int_{t}^{t+h}f_2(s)dsH_2\right)\exp\left(-\I\int_{t}^{t+h}f_1(s)dsH_1\right). 
\end{equation}
The second-order standard Trotter formula is 
\begin{equation} \label{eq:scheme_s_trotter_2nd}
    U_{s,2}(t+h,t) = \exp\left(-\frac{\I h}{2}f_1(t+h/2)H_1\right)\exp\left(-\I h f_2(t+h/2)H_2\right)\exp\left(-\frac{\I h}{2}f_1(t+h/2)H_1\right). 
\end{equation}
The second-order generalized Trotter formula is
\begin{equation}\label{eq:scheme_g_trotter_2nd}
    U_{g,2}(t+h,t) = \exp\left(-\I\int_{t+h/2}^{t+h}f_1(s)dsH_1\right)\exp\left(-\I\int_{t}^{t+h}f_2(s)dsH_2\right)\exp\left(-\I\int_{t}^{t+h/2}f_1(s)dsH_1\right). 
\end{equation}

It is clear that the difference between the standard and the generalized Trotter formulae lies in the temporal treatment of $f_1$ and $f_2$, and the standard Trotter formula can be viewed as applying certain quadrature rules in representing the integrals of $f_1$ and $f_2$. 
From now on we assume that $\int_a^b f(s) ds$ can be accurately 
computed with negligible extra cost for any scalar-valued smooth function $f(s)$. 
Furthermore, we remark that although in our definitions of the schemes we perform evolution governed by $H_1$ at first and then by $H_2$, the order of $H_1$ and $H_2$ only affects the absolute preconstants in the error bounds and will not lead to any difference in the asymptotic scalings.

\subsection{Error representations}

For the time-independent Trotter formula, the work of \cite{ChildsSuTranEtAl2020, DescombesThalhammer2010} prove a commutator type of error of any order by writing down an explicit error representation via variation of parameters formula. Here we follow the procedure in~\cite{ChildsSuTranEtAl2020} to write down the corresponding error representations for standard and generalized time-dependent Trotter formulae, which turn out to be the starting point for proving both the operator norm and the vector norm error bounds. 
Although we only present the error representation on the interval $[0,h]$, this is just for notation simplicity and with minor modifications the results naturally hold on $[t,t+h]$ for any $t$. 
The proofs are given in \cref{append:proof_error_represent}.

\begin{lemma}[Error representation of the first-order standard Trotter formula]\label{lem:error_rep_s_1st}
    \begin{equation}
        \begin{split}
            U_{s,1}(h,0) - U(h,0) = \int_0^{h} U(h,s) \exp\left(-\I s f_2(s)H_2\right)E_{s,1}(s)\exp\left(-\I s f_1(s)H_1\right) ds
        \end{split}
    \end{equation}
    where 
    \begin{equation}
        \begin{split}
            E_{s,1}(h) & = \int_{0}^h f_1(h)f_2(s) \left(\exp\left(\ad_{\I s f_2(s) H_2}\right)([H_1,H_2])\right)ds - \I h f_1'(h)H_1 - \I h f_2'(h)H_2 \\
            & \quad + \int_{0}^h  s f_1(h) f_2'(s) \left(\exp\left(\ad_{\I s f_2(s) H_2}\right)([H_1,H_2])\right)ds. 
        \end{split}
    \end{equation}
\end{lemma}

\begin{lemma}[Error representation of the first-order generalized Trotter formula]
\label{lem:error_rep_g_1st}
    \begin{equation}
        \begin{split}
            U_{g,1}(h,0) - U(h,0) = \int_0^{h} U(h,s) \exp\left(-\I  \int_0^s f_2(s')ds' H_2\right)E_{g,1}(s)\exp\left(-\I \int_0^s f_1(s')ds' H_1\right) ds
        \end{split}
    \end{equation}
    where 
    \begin{equation}
        \begin{split}
            E_{g,1}(h) &= \int_0^h f_1(h) f_2(s) \left(\exp\left(\ad_{\I \int_0^h f_2(s')ds' H_2}\right)([H_1,H_2])\right) ds. 
        \end{split}
    \end{equation}
\end{lemma}

\begin{lemma}[Error representation of the second-order standard Trotter formula]
\label{lem:error_rep_s_2nd}
    \begin{equation}
        \begin{split}
            U_{s,2}(h,0) - U(h,0) = \int_0^{h} &U(h,s) \exp\left(-\frac{\I s}{2}f_1(s/2)H_1\right)E_{s,2}(s)\\
            &\exp\left(-\I s f_2(s/2)H_2\right)\exp\left(-\frac{\I s}{2}f_1(s/2)H_1\right) ds
        \end{split}
    \end{equation}
    where $E_{s,2}$ is defined in \cref{eqn:Es2}.
\end{lemma}
\begin{lemma}[Error representation of the second-order generalized Trotter formula]
\label{lem:error_rep_g_2nd}
    \begin{equation}
        \begin{split}
            U_{g,2}(h,0) - U(h,0) = \int_0^{h} & U(h,s) \exp\left(-\I\int_{s/2}^s f_1(s')ds' H_1\right)E_{g,2}(s)\\
            &\exp\left(-\I \int_0^s f_2(s')ds' H_2\right)\exp\left(-\I\int_0^{s/2}f_1(s')ds' H_1\right) ds
        \end{split}
    \end{equation}
    where $E_{g,2}$ is defined in \cref{eqn:Eg2}.
\end{lemma}

The expressions of these exact error representations are somewhat complicated, but the structures for all the representations are the same. 
As introduced in \cref{sec:main_idea}, the error representations for both standard and generalized Trotter formulae of first and second-order are linear combinations of integrals with integrands in the form of \cref{eqn:trotter_error_rep_form}, which can be expressed as the multiplication of matrix exponentials, Hamiltonians, and commutators of Hamiltonians. 

Before we proceed, we remark that in the error representations of standard Trotter formulae of first and second-order, besides the $\Or(h^{p+1})$ terms, we also include higher order terms $\Or(h^{p+2})$ and/or $\Or(h^{p+3})$. 
This is because we aim at writing down the exact error terms, and the Taylor expansion of terms like $\exp(-\I h f_1(h) H_1)$ will naturally involve higher order term even though we only expand it up to the desired lower order. 
For example, if we look at the first-order derivative of $\exp(-\I h f_1(h) H_1)$, then 
$$\frac{d}{dh} (\exp(-\I h f_1(h) H_1)) = -\I f_1(h)H_1\exp(-\I h f_1(h) H_1) - \I h f_1'(h)H_1\exp(-\I h f_1(h) H_1).$$
We can observe that the first term is $\Or(1)$ and the second term is $\Or(h)$, which are on different scales. 
Therefore, the same order term in the Taylor expansion of the unitaries does not necessarily have the same scaling in terms of $h$. 

\begin{rem}[Exact error representation]It is possible to derive a simpler error bound by considering only the lowest order term and discarding all the higher order terms in $h$. However, such an error bound would not reveal the commutator scaling in the higher order remainder terms. For instance, for the time-independent Hamiltonian simulation,~\cite{Thalhammer2008} deduces an error bound for the $p$-th order Trotter formula, in which the $\Or(h^p)$ term has a commutator structure, but the higher order terms do not. This leads to complexity overhead when the spectral norms of the Hamiltonians become large.  The work~\cite{ChildsSuTranEtAl2020} fixes this issue by deriving an exact error representation, demonstrating the validity of the commutator scaling for high order terms as well. Therefore for time-dependent simulation, we also preserve all the terms in the error representation (at least for now). 
We can observe that all the terms in the exact representation, regardless of the order in $h$, are in the form of \cref{eqn:trotter_error_rep_form}, thus no overhead will be introduced by higher order terms and it is safe to bound them by the lowest order term later in estimating complexity. 
\end{rem}


\section{Operator norm error bounds}\label{sec:operator_norm}

We first establish the operator norm error bounds. 
In this section we assume that $H_1$ and $H_2$ are two bounded operators. 
The operator norm error bounds can be directly obtained from the error representations by bounding the operator norms of all the unitaries by 1. 
\begin{theorem}\label{thm:error_trotterstep}
The error of each standard/generalized Trotter step measured in the operator norm is as follows:
\begin{enumerate}
    \item First-order standard Trotter formula:
    \begin{equation}
        \|U_{s,1}(h,0)-U(h,0)\| \leq \alpha_{s,1} h^2 + \beta_{s,1} h^3,
    \end{equation}
    where 
    \begin{equation}
        \alpha_{s,1} = \frac{1}{2}\|f_1'\|_{\infty}\|H_1\| + \frac{1}{2}\|f_2'\|_{\infty}\|H_2\| + \frac{1}{2}\|f_1\|_{\infty}\|f_2\|_{\infty} \|[H_1,H_2]\|
    \end{equation}
    and 
    \begin{equation}
        \beta_{s,1} = \frac{1}{6} \|f_1\|_{\infty}\|f_2'\|_{\infty} \|[H_1,H_2]\|. 
    \end{equation}
    
    \item First-order generalized Trotter formula:
    \begin{equation}
        \|U_{g,1}(h,0) - U(h,0)\| \leq \alpha_{g,1} h^2
    \end{equation}
    where 
    \begin{equation}
        \alpha_{g,1} = \frac{1}{2}\|f_1\|_{\infty}\|f_2\|_{\infty} \|[H_1,H_2]\|. 
    \end{equation}
    
    \item Second-order standard Trotter  formula: 
    \begin{equation}
        \|U_{s,2}(h,0)-U(h,0)\| \leq \alpha_{s,2} h^3 + \beta_{s,2} h^4 + \gamma_{s,2} h^5,
    \end{equation}
    where 
    \begin{equation}
        \begin{split}
            \alpha_{s,2} &= \frac{7}{24}\|f_1''\|_{\infty}\|H_1\|  +\frac{1}{12} \|f_1'\|_{\infty}\|f_2\|_{\infty} \|H_1\|+  \frac{7}{24}\|f_2''\|_{\infty}\|H_2\|  \\
        & \quad + \frac{1}{6} (\|f_1'\|_{\infty}\|f_2\|_{\infty}+\|f_1\|_{\infty}\|f_2'\|_{\infty})\|[H_1,H_2]\| \\
        & \quad + \frac{1}{24}\|f_1\|_{\infty}^2\|f_2\|_{\infty}\|[H_1,[H_1,H_2]]\| + \frac{1}{12}\|f_1\|_{\infty}\|f_2\|_{\infty}^2\|[H_2,[H_1,H_2]]\| , 
        \end{split}
    \end{equation}
    \begin{equation}
        \begin{split}
            \beta_{s,2} &= \frac{1}{64} \|f_1'\|_{\infty}\|f_2'\|_{\infty}\|H_1\| \\
            & \quad + \left(\frac{1}{192}\|f_1\|_{\infty}\|f_2''\|_{\infty}+\frac{1}{192}\|f_1''\|_{\infty}\|f_2\|_{\infty} + \frac{1}{48}\|f_1'\|_{\infty}\|f_2'\|_{\infty}  \right)\|[H_1,H_2]\| \\
        & \quad + \frac{1}{96}\|f_1\|_{\infty} \|f_1'\|_{\infty} \|f_2\|_{\infty} \|[H_1,[H_1,H_2]]\| + \frac{1}{48} \|f_1\|_{\infty} \|f_2\|_{\infty} \|f_2'\|_{\infty} \|[H_2,[H_1,H_2]]\| ,
        \end{split}
    \end{equation}
    and 
    \begin{equation}
        \begin{split}
            \gamma_{s,2} = \frac{1}{960}\|f_1'\|_{\infty}^2\|f_2\|_{\infty}\|[H_1,[H_1,H_2]]\| + \frac{1}{480} \|f_1\|_{\infty}\|f_2'\|_{\infty}^2\|[H_2,[H_1,H_2]]\|. 
        \end{split}
    \end{equation}
    \item  Second-order generalized Trotter formula: 
    \begin{equation}
        \|U_{g,2}(h,0)-U(h,0)\| \leq \alpha_{g,2} h^3,
    \end{equation}
    where 
    \begin{equation}
        \begin{split}
            \alpha_{g,2} &= \left(\frac{7}{12}\|f_1\|_{\infty}\|f_2'\|_{\infty}+\frac{11}{24}\|f_1'\|_{\infty}\|f_2\|_{\infty}\right)\|[H_1,H_2]\| \\
            & \quad + \frac{3}{8}\|f_1\|_{\infty}^2\|f_2\|_{\infty}\|[H_1,[H_1,H_2]]\| + \frac{1}{12}\|f_1\|_{\infty}\|f_2\|_{\infty}^2\|[H_2,[H_2,H_1]]\|. 
        \end{split}
    \end{equation}
\end{enumerate}
    
\end{theorem}

\begin{rem}[The preconstants in \cref{thm:error_trotterstep}]First, the preconstants of the standard Trotter formula involve the norms of both the Hamiltonians as well as their commutators, while the preconstants of the generalized Trotter formula of the first order only involve the commutator $[H_1,H_2]$, and those of the second order involve further nested commutators. Second, the $p$-th order standard Trotter scheme ($p = 1, 2$) depends on the $p$-th order derivatives of the control functions while the $p$-th generalized Trotter scheme depends only on $(p-1)$-th order derivatives. In this sense, when the time derivatives of $f_1$ and $f_2$ are large (or when the regularity of $f_1,f_2$ are limited), the generalized Trotter formula can further outperform the standard Trotter method.
\end{rem}

Now we move on to the global error bounds. 
To obtain an approximation of the exact unitary evolution up to time $T$, we can divide the time interval $[0,T]$ into $L$ equilength segments and implement Trotter discretization on each segment. 
Since the evolutions for both continuous and discretized cases are unitary, the global error is simply a linear accumulation of local errors at each time step. 
For sufficiently large $L$, the total error can be controlled to be arbitrarily small. 
\begin{theorem}\label{thm:operator_norm_bound}
    Let $T > 0$ be the evolution time, and the dynamics \cref{eqn:td_hamsim} is discretized via standard and generalized Trotter formulae with $L$ equidistant time steps (thus the time step size $h=T/L$). Then 
        \begin{align}
            &\left\|\prod_{l=1}^{L}U_{s,1}\left(\frac{lT}{L},\frac{(l-1)T}{L}\right) - U(T,0)\right\| \leq \alpha_{s,1} \frac{T^2}{L} + \beta_{s,1} \frac{T^3}{L^2},\\
            &\left\|\prod_{l=1}^{L}U_{g,1}\left(\frac{lT}{L},\frac{(l-1)T}{L}\right) - U(T,0)\right\| \leq \alpha_{g,1} \frac{T^2}{L},\\ 
            &\left\|\prod_{l=1}^{L}U_{s,2}\left(\frac{lT}{L},\frac{(l-1)T}{L}\right) - U(T,0)\right\| \leq \alpha_{s,2} \frac{T^3}{L^2} + \beta_{s,2} \frac{T^4}{L^3} + \gamma_{s,2} \frac{T^5}{L^4},\\
          &\left\|\prod_{l=1}^{L}U_{g,2}\left(\frac{lT}{L},\frac{(l-1)T}{L}\right) - U(T,0)\right\| \leq \alpha_{g,2} \frac{T^3}{L^2},
        \end{align}
    where preconstants $\alpha$ and $\beta$ are defined in \cref{thm:error_trotterstep}. 
\end{theorem}



\section{Vector norm error bounds}\label{sec:vector_norm}

Now we consider the case when $H_1$ is an approximation of an unbounded operator, while $H_2$ remains reasonably bounded. 
In this section, we assume that $\norm{H_1}\gg \norm{H_2}$, and the functions $f_1$, $f_2$, as well as their first and second-order derivatives are bounded. To simplify the proof and emphasize our focus on overcoming the difficulty brought by $H_1$, throughout this section we will not track the explicit dependence on $H_2$, $f_1$ and $f_2$.  We use the notation $\widetilde{C}$ with a tilde above to denote preconstants (with possibly varying sizes in different inequalities) which can depend polynomially on $H_2$, $\|f_1^{(k)}\|_{\infty}$ and $\|f_2^{(k)}\|_{\infty}$ but do not depend on $H_1$. Furthermore, we make the following assumptions. 

\begin{assump}[Bounds of commutators]\label{assump:commutator_bound}
We assume $H_1$ is a positive semidefinite operator, and there exists an operator $D_1$ such that $H_1=D_1^\dagger D_1$. 
Furthermore, we assume for any vector $\vec{v}$, there exist constants $\widetilde{C}_{1},\widetilde{C}_{2}$ such that 
    \begin{equation} \label{eqn:assump_comm1}
        \|[H_1,H_2]\vec{v}\| \leq \widetilde{C}_1 (\|D_1\vec{v}\| + \|\vec{v}\|),
    \end{equation}
    and 
    \begin{equation} \label{eqn:assump_comm2}
        \|[H_1,[H_1,H_2]]\vec{v}\| \leq \widetilde{C}_2 (\|H_1\vec{v}\| + \|\vec{v}\|). 
    \end{equation}
\end{assump}

\cref{assump:commutator_bound} has been used in previous works~\cite{JahnkeLubich2000,HochbruckLubich2003} related to the vector norm error bounds of time-independent Trotter formula and exponential integrators for $H=-\Delta+V(x)$, where $H_1=-\Delta$ is positive semidefinite and $H_2=V(x)$ is a bounded operator.
It will be helpful to understand \cref{assump:commutator_bound} from a continuous analog, in which $H_1=D_1^\dagger D_1$ and $D_1$ is a first-order differential operator. 
A direct calculation shows that the operator 
$$[-\Delta,V] = - \partial_x^2 V - 2 (\partial_x V)\partial_x$$ 
is a first-order differential operator, and 
\begin{equation}\label{eqn:comm_explicit_laplapV}
[-\Delta, [-\Delta, V]]  = \partial_x^4 V + 4 \partial_x^3 V \partial_x + 4 \partial_x^2 V \partial_x^2
\end{equation}
is a second-order differential operator, given that $V(x)$ is a $C^4$ function. 
These are exactly what \cref{eqn:assump_comm1,eqn:assump_comm2} are addressing. Furthermore, if the wavefunction $v$ is smooth enough with bounded derivatives, then the right hand sides of these inequalities are bounded, which provides the key motivation and possibility to establish vector norm error bounds and obtain improvement in complexity estimates. In the context of quantum simulation, since all matrices and vectors are finite dimensional, we omit the explicit statements of regularity assumptions of $v$ below. 


As we have briefly discussed before, starting from the error representations, the vector norm error bounds are just linear combinations of the terms in the form of \cref{eqn:trotter_error_rep_form_vec}, and the key step to prove vector norm error bounds is to exchange the order of a Hamiltonian or an commutator with matrix exponentials. 
We find that such an exchange of order will not introduce any overhead in the error bounds. 
This is established by the following lemma. 

\begin{lemma}\label{lem:vector_norm}
    Under \cref{assump:commutator_bound}, we have the following:
    \begin{enumerate}
        \item For any vector $\vec{v}$, 
        \begin{equation}\label{eqn:D1v_bound}
            \|D_1\vec{v}\| \leq \|H_1\vec{v}\| + \|\vec{v}\|. 
        \end{equation}
      \item Let $\xi$ be any real number such that $(\widetilde{C}_1+\|H_2\|) |\xi|  \leq 1/2$. Then for any vector $\vec{v}$, 
        \begin{equation}\label{eqn:H1v_bound1}
        \|H_1 \exp\left(\I \xi H_2\right)\vec{v}\| \leq 2 (\|H_1 \vec{v}\| + \|\vec{v}\|).
        \end{equation}
        \item Let $K$ be a positive integer, and $H_{l_{k}}$ be either $H_1$ or $H_2$, and $\xi_k$ be some real numbers for $1\le k\le K$. Assume that $(\widetilde{C}_1+\|H_2\|) |\xi_k|h  \leq 1/2$, then for any vector $\vec{v}$, all the following inequalities hold: 
        \begin{equation}\label{eqn:lem_3}
            \begin{split}
                \left\|H_1  \left[\prod_{k=1}^{K}\exp\left(\I h \xi_{k} H_{l_{k}}\right) \right]\vec{v}\right\| & \leq \widetilde{C}(\|H_1\vec{v}\| + \|\vec{v}\|), \\
                \left\|H_2  \left[\prod_{k=1}^{K}\exp\left(\I h \xi_{k} H_{l_{k}}\right) \right]\vec{v}\right\| & \leq \widetilde{C}\|\vec{v}\|,\\
                 \left\|[H_1,H_2]  \left[\prod_{k=1}^{K}\exp\left(\I h \xi_{k} H_{l_{k}}\right) \right]\vec{v}\right\| & \leq \widetilde{C} \left(\sqrt{\|\vec{v}\|\|H_1\vec{v}\|} + \|\vec{v}\|\right),\\
                \left\|[H_1,[H_1,H_2]]  \left[\prod_{k=1}^{K}\exp\left(\I h \xi_{k} H_{l_{k}}\right) \right]\vec{v}\right\| & \leq \widetilde{C}(\|H_1\vec{v}\| + \|\vec{v}\|),\\
                \left\|[H_2,[H_2,H_1]]   \left[\prod_{k=1}^{K}\exp\left(\I h \xi_{k} H_{l_{k}}\right) \right]\vec{v}\right\| & \leq \widetilde{C}(\|H_1\vec{v}\| + \|\vec{v}\|). 
            \end{split}
        \end{equation}
        for some constant $\widetilde{C}>0$ depending only on $\|H_2\|$. 
    \end{enumerate}
\end{lemma}
\begin{proof}
    1. By the definition of $D_1$ and the Cauchy-Schwarz inequality,
    \begin{equation}\label{eqn:D1v_step}
            \|D_1\vec{v}\|^2 = (D_1\vec{v})^{\dagger}D_1\vec{v} =  \vec{v}^{\dagger} H_1 \vec{v}
             \leq \|\vec{v}\|\|H_1 \vec{v}\| 
             \leq (\|H_1 \vec{v}\| + \|\vec{v}\|)^2. 
    \end{equation}
    
      2. We start with Taylor's theorem of $\exp(\I t \xi  H_2)$ up to first-order, then 
    \begin{equation*}
    \begin{split}
        H_1\exp(\I t \xi  H_2 ) \vec{v} &= H_1 \vec{v} + \int_0^t \I \xi H_1H_2 \exp\left(\I  \alpha \xi H_2\right) \vec{v} d\alpha \\
        &= H_1 \vec{v} + \int_0^t \I \xi [H_1,H_2] \exp\left(\I \alpha \xi H_2\right) \vec{v} d\alpha + \int_0^t \I  \xi H_2H_1 \exp\left(\I \alpha \xi H_2\right) \vec{v} d\alpha.
    \end{split}
    \end{equation*}
    The norm can be estimated using \cref{eqn:assump_comm1} as
    \begin{equation} \label{eq: lemma5_pf_h1exph2}
    \begin{split}
        \left\| H_1\exp(\I t\xi  H_2 ) \vec{v} \right\| 
        &\leq  \| H_1 \vec{v}\| + \int_0^t |\xi|\| [H_1,H_2] \exp\left(\I \alpha \xi H_2\right)\vec{v}\| d\alpha + \int_0^t |\xi| \| H_2\| \left\| H_1 \exp\left(\I \alpha \xi H_2\right)\vec{v} \right\| d\alpha \\
        &\leq \| H_1\vec{v}\| + \widetilde{C}_1\|\xi\|\|\vec{v}\| t + \int_0^t \widetilde{C}_1|\xi|\| D_1 \exp\left(\I \alpha \xi H_2\right)\vec{v}\| d\alpha  \\
        & \quad + \int_0^t |\xi| \| H_2\| \left\| H_1 \exp\left(\I \alpha \xi H_2\right)\vec{v} \right\| d\alpha.
    \end{split}
    \end{equation}
Define $M(t) : = \left\| H_1\exp(\I t\xi  H_2 ) \vec{v} \right\|$, and it follows from \cref{eqn:D1v_bound} that
    \begin{equation*}
    \begin{split}
    \| D_1 \exp\left(\I t \xi H_2\right)\vec{v}\| \leq M(t) + \|\vec{v}\|. 
    \end{split}
    \end{equation*}
    Thus \cref{eq: lemma5_pf_h1exph2} can be rewritten as
    \begin{equation*}
        M(t) \leq \|H_1 \vec{v}\| + 2\widetilde{C}_1 |\xi| \| \vec{v}\| t + \int_0^t |\xi| \left(\widetilde{C}_1 +\|H_2\| \right) M(\alpha) \, d\alpha.
    \end{equation*}
 Since $\|H_1 \vec{v}\| + 2\widetilde{C}_1 |\xi| \| \vec{v}\| t$ is non-decreasing with respect to $t$, applying Gronwall's inequality yields the bound
 \begin{equation*}
 M(t) \leq \left(\|H_1 \vec{v}\| + 2\widetilde{C}_1 |\xi| \| \vec{v}\| t\right) \exp\left( |\xi| (\widetilde{C}_1 +\|H_2\| ) t \right).
 \end{equation*}
 Finally taking $t = 1$ and applying the condition on $\xi$, the desired result is achieved
 \begin{equation*}
\left\| H_1\exp(\I \xi  H_2 )\vec{v} \right\|  \leq \left(\|H_1 \vec{v}\| + 2\widetilde{C}_1 |\xi| \| \vec{v}\| \right) \exp\left( |\xi| (\widetilde{C}_1 +\|H_2\| )  \right) \leq 2(\|H_1 \vec{v}\| +  \|\vec{v}\|). 
 \end{equation*}  

 3. We first show that it suffices to only prove the first inequality in \cref{eqn:lem_3}. 
 Let $\vec{w} = \left[\prod_{k=1}^{K}\exp\left(\I h \xi_{k} H_{l_{k}}\right) \right]\vec{v}$. Since $H_2$ is bounded, $\|H_2 \vec{w}\|$ is directly bounded by $\widetilde{C}\|\vec{v}\|$ . 
By \cref{eqn:assump_comm1,eqn:D1v_step}, and the fact $\norm{\vec{w}}=\norm{\vec{v}}$, we have
 \begin{equation}
     \|[H_1,H_2]\vec{w}\| \leq \widetilde{C}_1(\|D_1\vec{w}\| + \|\vec{w}\|) 
      \leq \widetilde{C}\left(\sqrt{\|\vec{w}\|\|H_1\vec{w}\|}+\|\vec{w}\|\right) 
      = \widetilde{C}\left(\sqrt{\|\vec{v}\|\|H_1\vec{w}\|}+\|\vec{v}\|\right). 
 \end{equation}
 Similarly \cref{eqn:assump_comm2} gives 
 \begin{equation}
     \|[H_1,[H_1,H_2]]\vec{w}\| \leq \widetilde{C}_2(\|H_1\vec{w}\| + \|\vec{w}\|) = \widetilde{C}_2(\|H_1\vec{w}\| + \|\vec{v}\|).
 \end{equation}
 Furthermore, 
 \begin{equation}
 \begin{split}
     \|[H_2,[H_2,H_1]]\vec{w}\| &\leq \|H_2[H_2,H_1]\vec{w}\| + \|[H_2,H_1]H_2\vec{w}\| \\
     & \leq \widetilde{C} \|[H_2,H_1]\vec{w}\| + \widetilde{C}_1(\|D_1H_2\vec{w}\| + \|H_2\vec{w}\|) \\
     & \leq \widetilde{C}(\|D_1\vec{w}\| + \|\vec{w}\|) + \widetilde{C}(\|H_1H_2\vec{w}\|+\|H_2\vec{w}\|) \\
     & \leq \widetilde{C}(\|H_1\vec{w}\|+\|\vec{w}\|) + \widetilde{C}(\|[H_1,H_2]\vec{w}\|+\|H_2H_1\vec{w}\| + \|\vec{w}\|) \\
     & \leq \widetilde{C} (\|H_1\vec{w}\|+\|\vec{w}\|) + \widetilde{C} (\|D_1\vec{w}\| + \|\vec{w}\|)\\
     & \leq \widetilde{C} (\|H_1\vec{w}\|+\|\vec{w}\|)  = \widetilde{C} (\|H_1\vec{w}\|+\|\vec{v}\|). 
 \end{split}
 \end{equation}
 Therefore we only need to bound $\|H_1\vec{w}\|$ further by $\widetilde{C}(\|H_1\vec{v}\|+\|\vec{v}\|)$. 
 
 Notice that $\|H_1\exp(\I\xi H_1)\vec{v}\| = \|\exp(\I\xi H_1)H_1\vec{v}\| = \|H_1\vec{v}\| $, together with \cref{eqn:H1v_bound1},  
 \begin{equation}
     \|H_1\exp(\I\xi H_l)\vec{v}\| \leq 2(\|H_1\vec{v}\| + \|\vec{v}\|)
 \end{equation}
 for $H_l$ being either $H_1$ or $H_2$. 
 Then we have the recursive relation
 \begin{equation}
     \begin{split}
         \left\|H_1\left[\prod_{k=1}^{K}\exp\left(\I h \xi_{k} H_{l_{k}}\right) \right]\vec{v}\right\| & \leq 2\left\|H_1\left[\prod_{k=1}^{K-1}\exp\left(\I h \xi_{k} H_{l_{k}}\right) \right]\vec{v}\right\| + 2\left\|\left[\prod_{k=1}^{K-1}\exp\left(\I h \xi_{k} H_{l_{k}}\right) \right]\vec{v}\right\| \\
         & = 2\left\|H_1\left[\prod_{k=1}^{K-1}\exp\left(\I h \xi_{k} H_{l_{k}}\right) \right]\vec{v}\right\| + 2\left\|\vec{v}\right\|.
     \end{split}
 \end{equation}
 By applying this estimation $K$ times, we obtain the desired result 
 \begin{equation}
         \left\|H_1\left[\prod_{k=1}^{K}\exp\left(\I h \xi_{k} H_{l_{k}}\right) \right]\vec{v}\right\| \leq \widetilde{C} \left(\|H_1\vec{v}\| + \|\vec{v}\|\right). 
 \end{equation}

\end{proof}

Now we are ready to state our main theorems for the vector norm estimates.

\begin{theorem}\label{thm:Trotter_vector_norm_error_bound_local}
    For any vector $v$ and time step size $h \leq (\|f_1\|_{\infty}+\|f_2\|_{\infty})^{-1}(\widetilde{C}_1+\|H_2\|)^{-1}/2$, there exists a constant $\widetilde{C}$ such that 
    \begin{align*}
        \|U_{s,1}(h,0)\vec{v}-U(h,0)\vec{v}\| & \leq \widetilde{C}h^2(\|H_1\vec{v}\| + \|\vec{v}\|), \\ 
        \|U_{g,1}(h,0)\vec{v}-U(h,0)\vec{v}\| & \leq \widetilde{C}h^2(\sqrt{\|\vec{v}\|\|H_1\vec{v}\|} + \|\vec{v}\|), \\
        \|U_{s,2}(h,0)\vec{v}-U(h,0)\vec{v}\| & \leq \widetilde{C}h^3(\|H_1\vec{v}\| + \|\vec{v}\|), \\
        \|U_{g,2}(h,0)\vec{v}-U(h,0)\vec{v}\| & \leq \widetilde{C}h^3(\|H_1\vec{v}\| + \|\vec{v}\|). 
    \end{align*}
\end{theorem}
\begin{proof}
    We start with the error representations (\cref{lem:error_rep_s_1st,lem:error_rep_g_1st,lem:error_rep_s_2nd,lem:error_rep_g_2nd}). 
    By multiplying a vector $\vec{v}$ on the right of the error representations, bounding all the $f_j^{(k)}$ by its supremum, bounding all the higher order terms involving the (nested) commutators by second-order terms, and bounding all the unitaries multiplied on the left by $1$, we obtain  
    \begin{align*}
        \|U_{s,1}(h,0)\vec{v}-U(h,0)\vec{v}\| & \leq \widetilde{C} \theta_{s,1} h^2, \\
        \|U_{g,1}(h,0)\vec{v}-U(h,0)\vec{v}\| & \leq \widetilde{C} \theta_{g,1} h^2, \\
        \|U_{s,2}(h,0)\vec{v}-U(h,0)\vec{v}\| & \leq \widetilde{C} \theta_{s,2} h^3, \\
        \|U_{g,2}(h,0)\vec{v}-U(h,0)\vec{v}\| & \leq \widetilde{C} \theta_{g,2} h^3, 
    \end{align*}
    where $\theta_s$ and $\theta_g$ are linear combinations of the terms in the form of \cref{eqn:lem_3}, with an exception that $\theta_{g,1}$ only consists terms like $\|[H_1,H_2]\vec{w}\|$. 
    Then part 3 of \cref{lem:vector_norm} completes the proof. 
\end{proof}

Similar to the operator norm error bound case, from \cref{thm:Trotter_vector_norm_error_bound_local}, we can establish the global error bound and estimate the total number of time steps we need to achieve a desired accuracy using standard and generalized Trotter formulae. The proof of the global error bound is essentially the same as the standard argument for error accumulation in quantum computing, that is, to replace the exact evolution operator by numerical evolution operator step by step and bound each step by local error bound. 
In the operator norm case, it does not matter whether we replace the local evolution operator in a forward or backward fashion. 
However, in the vector norm case, the order of the replacements indeed matters. In particular, we would like to obtain an error bound that depends on the exact, instead of the numerical solution of the dynamics. 
We state our vector norm global error bound in \cref{thm:vector_norm_bound}, and provide a complete proof for the second-order generalized Trotter formula. 
\begin{theorem}\label{thm:vector_norm_bound}
    Let $T > 0$ be the evolution time, $\vec{\psi}(t)$ be the exact solution of the dynamics \cref{eqn:td_hamsim}, and the dynamics \cref{eqn:td_hamsim} is discretized via standard and generalized Trotter formulae with $L$ time steps such that the time step size $h=T/L$ is bounded by $(\|f_1\|_{\infty}+\|f_2\|_{\infty})^{-1}(\widetilde{C}_1+\|H_2\|)^{-1}/2$. 
    Then there exists a constant $\widetilde{C}$ such that 
    \begin{align}
        &\left\|\left(\prod_{l=1}^L U_{s,1}\left(\frac{lT}{L},\frac{(l-1)T}{L}\right)\right)\vec{\psi}(0) - U(T,0)\vec{\psi}(0)\right\| \leq \widetilde{C} \frac{T^2}{L} \left(\sup_{t\in[0,T]}\|H_1\vec{\psi}(t)\|+\|\vec{\psi}(0)\|\right),\\
        &\left\|\left(\prod_{l=1}^L U_{g,1}\left(\frac{lT}{L},\frac{(l-1)T}{L}\right)\right)\vec{\psi}(0) - U(T,0)\vec{\psi}(0)\right\| \leq \widetilde{C} \frac{T^2}{L} \left(\sup_{t\in[0,T]}\sqrt{\|\vec{\psi}(0)\|\|H_1\vec{\psi}(t)\|}+\|\vec{\psi}(0)\|\right),\\
        &\left\|\left(\prod_{l=1}^L U_{s,2}\left(\frac{lT}{L},\frac{(l-1)T}{L}\right)\right)\vec{\psi}(0) - U(T,0)\vec{\psi}(0)\right\| \leq \widetilde{C} \frac{T^3}{L^2} \left(\sup_{t\in[0,T]}\|H_1\vec{\psi}(t)\|+\|\vec{\psi}(0)\|\right),\\
        &\left\|\left(\prod_{l=1}^L U_{g,2}\left(\frac{lT}{L},\frac{(l-1)T}{L}\right)\right)\vec{\psi}(0) - U(T,0)\vec{\psi}(0)\right\| \leq \widetilde{C} \frac{T^3}{L^2} \left(\sup_{t\in[0,T]}\|H_1\vec{\psi}(t)\|+\|\vec{\psi}(0)\|\right).\label{eqn:vectornorm_ug2}
    \end{align}
\end{theorem}

\paragraph{Remark} 
Although the error for the first-order generalized Trotter formula can also be bounded by $\sup_{t\in[0,T]}\|H_1\vec{\psi}(t)\|+\|\vec{\psi}(0)\|$ as the other schemes by a direct application of the Cauchy-Schwarz inequality, we keep it as $\sup_{t\in[0,T]}\sqrt{\|\vec{\psi}(0)\|\|H_1\vec{\psi}(t)\|}+\|\vec{\psi}(0)\|$ which promotes a better dependence on $\|H_1\vec{\psi}(t)\|$. This improvement is achievable only for the first-order generalized Trotter formula  since its error only contains terms depending on $[H_1, H_2]$ which process a better bound as in \cref{eqn:lem_3}, while the other schemes also contains terms depending on $H_1$ and/or the nested commutators $[H_1, [H_1, H_2]]$ and $[H_1, [H_2, H_1]]$.
\begin{proof}

    Here we only present the proof for second-order generalized Trotter formula \cref{eqn:vectornorm_ug2}. 
    The other three cases can be proved using the same approach. 
    
    For the second-order generalized Trotter formula, according to \cref{thm:Trotter_vector_norm_error_bound_local} and notice that $\|\vec{\psi}(t)\|=\|\vec{\psi}(0)\|$ for all $t\in[0,T]$, we obtain
    \begin{equation}
        \begin{split}
            &\quad \left\|\left(\prod_{l=1}^L U_{g,2}\left(\frac{lT}{L},\frac{(l-1)T}{L}\right)\right)\vec{\psi}(0) - U(T,0)\vec{\psi}(0)\right\| \\
            & = \left\|\left(\prod_{l=1}^L U_{g,2}\left(\frac{lT}{L},\frac{(l-1)T}{L}\right)\right)\vec{\psi}(0) - \left(\prod_{l=1}^L U\left(\frac{lT}{L},\frac{(l-1)T}{L}\right)\right)\vec{\psi}(0)\right\|\\
            & \leq \sum_{k=1}^{L}\Bigg\|\left(\prod_{l=k+1}^L U_{g,2}\left(\frac{lT}{L},\frac{(l-1)T}{L}\right)\right)\left(\prod_{l=1}^{k} U\left(\frac{lT}{L},\frac{(l-1)T}{L}\right)\right)\vec{\psi}(0) \\
            & \quad\quad\quad\quad - \left(\prod_{l=k}^L U_{g,2}\left(\frac{lT}{L},\frac{(l-1)T}{L}\right)\right)\left(\prod_{l=1}^{k-1} U\left(\frac{lT}{L},\frac{(l-1)T}{L}\right)\right)\vec{\psi}(0)\Bigg\|\\
            &  \leq \sum_{k=1}^{L}\left\|\left(\prod_{l=1}^{k} U\left(\frac{lT}{L},\frac{(l-1)T}{L}\right)\right)\vec{\psi}(0) 
            -  U_{g,2}\left(\frac{kT}{L},\frac{(k-1)T}{L}\right)\left(\prod_{l=1}^{k-1} U\left(\frac{lT}{L},\frac{(l-1)T}{L}\right)\right)\vec{\psi}(0)\right\| \\
            & = \sum_{k=1}^{L}\left\|\left( U\left(\frac{kT}{L},\frac{(k-1)T}{L}\right) - U_{g,2}\left(\frac{kT}{L},\frac{(k-1)T}{L}\right)\right) 
            \vec{\psi}((k-1)T/L)\right\| \\
            & \leq \widetilde{C} \frac{T^3}{L^3} \sum_{k=1}^{L}\left(\left\|H_1\vec{\psi}((k-1)T/L)\right\| + \left\|\vec{\psi}((k-1)T/L)\right\|\right) \\
            & \leq \widetilde{C} \frac{T^3}{L^2} \left(\sup_{t\in[0,T]}\left\|H_1\vec{\psi}(t)\right\| + \left\|\vec{\psi}(0)\right\|\right). 
        \end{split}
    \end{equation}

\end{proof}


Now we compare the error bounds in terms of operator norm  (\cref{thm:operator_norm_bound}) with those in terms of vector norm (\cref{thm:vector_norm_bound}). 
We notice that the scalings with respect to $T$ and $L$ are the same for schemes of the same order, and the difference is in the dependence of the preconstants on $H_1$ and $H_2$. 
More precisely, the operator norm error bounds still depend on $\|H_1\|$ for standard Trotter formula and \REV{depend} on norms of commutators like $\|[H_1,[H_1,H_2]]\|$ for generalized Trotter formula. 
On the other hand, the vector norm error bounds only depend on $\norm{H_2}$, and the dependence on $H_1$ only appears in the form of $\|H_1\vec{\psi}\|$. 
Such a difference implies that the vector norm bounds can be much sharper than the operator norm bounds when $\norm{H_1}$ \REV{is} very large, but $\norm{H_1\vec{\psi}}$ and $\norm{H_2}$ are relatively small. 
We will show later that this is indeed the case for the model of interest in \cref{eqn:schrodinger_tdmass}.
Furthermore, the difference in the error bounds can influence the scaling of total required Trotter steps with respect to  the accuracy $\epsilon$. 

\section{Application to Schr\"odinger equation with time-dependent effective mass and frequency}\label{sec:harmonic}

The model of the Schr\"odinger equation with a time-dependent effective mass and frequency in \cref{eqn:schrodinger_tdmass} has been studied in many works \cite{DantasPedrosaEtAl1992,PedrosaGuedes1997,Pedrosa1997,JiKimEtAl1995,Feng2001,SchulzeHalberg2005}. 
Our goal is to study the complexity to obtain an $\epsilon$-approximation of the wavefunction at time $T\sim \Or(1)$, where $D=[0,1]$ with periodic boundary conditions. Throughout the section we make the following assumptions. 
\begin{enumerate}
    \item $M_{\text{eff}}(t)$ is positive function, and is uniformly bounded from below.

    \item $M_{\text{eff}}(t),\omega(t)$ are second-order continuously differentiable functions in $t$ with uniformly bounded function and derivative values up to second order. 
    \item $V(x)$ is a fourth-order continuously differentiable function in $x$ with bounded function and derivative values up to fourth order.
\end{enumerate}
Here the fourth order derivative of $V(x)$ is required when estimating the errors of the second-order formulae in the operator norm. To be specific, it guarantees the nested commutator $[H_1, [H_1, H_2]]$ in its spatial discretization with $n$ spatial grids to have an operator norm bounded by $n^2$ (instead of $n^4$). Without going into details of the discretization, which will be presented in the proofs, here we provide an intuition of this requirement on the continuous level -- the presence of the fourth derivative of $V$ in $[-\Delta, [-\Delta, V]]$ as in \cref{eqn:comm_explicit_laplapV}.
Since the control functions and potential are bounded, throughout we will not track explicitly the dependence on them and absorb them into the preconstant denoted by $\widetilde{C}$ or the big-O notation $\Or$ in our estimates. 

We discretize the dynamics \cref{eqn:schrodinger_tdmass} as follows. 
First we perform spatial discretization using a central finite difference scheme with $n$ equidistant nodes $x_k = k/n, 0 \leq k \leq n-1$. 
Then the semi-discretized dynamics becomes $\I \partial_t \vec{\psi}(t) = H(t)\vec{\psi}(t)$. 
Here the $k$-th entry of $\vec{\psi}(t)$ will be an approximation of the exact wavefunction evaluated at $t$ and $x = (k-1)/n$. 
$H(t) = f_1(t)H_1 + f_2(t)H_2$ with 
\begin{equation}\label{eqn:H_1_fd}
    H_1 = n^2 \left(\begin{array}{ccccc}
        2 & -1 & & & -1 \\
         -1& 2 & -1 & & \\
          & \ddots& \ddots& \ddots& \\
           & & -1& 2 & -1\\
        -1& & & -1 & 2 \\
    \end{array}\right), 
\end{equation}
and 
\begin{equation}\label{eqn:H_2_fd}
    H_2 = \text{diag}(V(0),V(1/n),\cdots,V((n-1)/n)).
\end{equation}
The standard or generalized Trotter formulae are used to discretize the dynamics in time with equidistant time steps and obtain numerical approximation of the wavefunction.

Furthermore, the $H_1$ and $H_2$ under central finite difference scheme using $n$ equidistant nodes satisfy \cref{assump:commutator_bound} with $H_1 = D_1^{\dagger}D_1$, 
\begin{equation} \label{eqn:d_1_fd}
    D_1 =  n \left(\begin{array}{ccccc}
        -1 &  & & & 1 \\
         1& -1 &  & & \\
          & \ddots& \ddots& & \\
           & & 1& -1 & \\
         & & & 1 & -1 
    \end{array}\right), 
\end{equation}
and therefore the vector norm error bounds proved in \cref{sec:vector_norm} can be applied. 
This can be verified by straightforward but somewhat tedious matrix computations. We formally state the result in \cref{lem:harmonic_verification_assumption}, and its proof is given in \cref{append:proof_assumption}.

\begin{lemma}\label{lem:harmonic_verification_assumption}
    Consider $H_1$ and $H_2$ defined in \cref{eqn:H_1_fd,eqn:H_2_fd}, 
    then \cref{assump:commutator_bound} is satisfied under the rescaled 2-norm $\|\cdot\|_{\star}$ with $D_1$ defined in \cref{eqn:d_1_fd}. 
\end{lemma}

\cref{lem:harmonic_verification_assumption} can be directly used to provide the scaling of the Hamiltonians and their commutators in terms of $n$. 
Using the fact that the matrix 2-norms can be estimated by considering its 1-norm (the maximum absolute column sum) and $\infty$-norm (the maximum absolute row sum) by virtue of the fact that $\| M\|_2 \leq \sqrt{\|M\|_1 \|M\|_\infty}$, we obtain the following Lemma. 
Notice that this result is consistent with the continuous picture that their continuous analogs $[-\Delta, V]$ and $[V, [-\Delta, V]]$ are differential operators of the first order and  $[-\Delta, [-\Delta, V]]$ of second order. 
\begin{lemma}\label{lem:harmonic_H_commu_scaling}
    Consider $H_1$, $H_2$ and $D_1$ defined in \cref{eqn:H_1_fd,eqn:H_2_fd,eqn:d_1_fd}, 
    then
    \begin{equation} \label{eqn:h1h2_scaling_in_n}
    \|H_1\| = \Or(n^2), \quad \|H_2\| = \Or(1), \quad \|D_1\| = \Or(n),
    \end{equation}
    and 
    \begin{equation} \label{eqn:commutator_scaling_in_n}
     \|[H_1, H_2]\| = \Or(n), \quad \|[H_2, [H_1, H_2]] \|\le 2\|[H_2\|\| [H_1, H_2]] \|= \Or(n), \quad \|[H_1, [H_1, H_2]] \|= \Or(n^2). 
     \end{equation}
\end{lemma} 

\begin{rem}
    Although $\|H_1\|$ depends quadratically in $n$, the time-independent simulations $\exp(-\I H_1)$ and $\exp(-\I H_2)$ can still be performed efficiently. For $H_1$, it can be diagonalized under Fourier transform. Specifically, let $F = (\omega^{jk}/\sqrt{n})_{j,k=0}^{n-1}$ be the Fourier transform unitary matrix with $\omega = \exp(2\pi\I/n)$, then $H_1 = F\Lambda F^{\dagger}$ where $\Lambda = \text{diag}(2n^2(1-\cos(2\pi\I j/n)))_{j=0}^{n-1}$. 
    Therefore $\exp(-\I H_1) = F \exp(-\I \Lambda) F^{\dagger}$ can be simulated efficiently, by first applying inverse quantum Fourier transform $F^{\dagger}$, then applying fast-forwarding techniques~\cite{childs2003exponential,ahokas2004improved} for $\exp(-\I \Lambda)$, and finally applying quantum Fourier transform $F$.
    For $H_2$, it can be implemented via either QSP technique~\cite{LowChuang2017} due to the boundedness of $\|H_2\|$, or fast-forwarding techniques since $H_2$ is a diagonal matrix as well. 
   Due to the efficiency of simulating $\exp(-\I H_1)$ and $\exp(-\I H_2)$, it is reasonable to estimate the query complexity by counting the total number of Trotter steps. 
\end{rem}

Now we are ready to analyze the errors. We measure the discretization errors using rescaled 2-norm, \emph{i.e.} $\|\vec{\widetilde{\psi}}(t)-(\phi(t,k/n))_{k=0}^{n-1}\|_{\star}$, where $\vec{\widetilde{\psi}}(t)$ is the numerical solution after spatial and time discretization at time $t$, and $\phi(t,x)$ denotes the exact solution. 
As discussed before, the reason why we use rescaled 2-norm, rather than regular 2-norm, to measure the error is because the exact solution $(\phi(t,k/n))_{k=0}^{n-1}$ is a discrete representation of the function $\phi$, which is normalized under continuous $L^2$ norm rather than discrete 2-norm. Furthermore, if we encode the spatial discretized solution $\vec{\psi}$ into a quantum state, then the normalized condition requires $\ket{\psi} \sim \frac{1}{\sqrt{n}}\vec{\psi}$, thus $\|\ket{\psi}\| \sim \frac{1}{\sqrt{n}}\|\vec{\psi}\| = \|\vec{\psi}\|_{\star}$, that is, under correct normalization in each scenario, bounding regular 2-norm error for quantum states is equivalent to bounding rescaled 2-norm error for spatial discretized vectors. We remark that if we would like to control the relative error, then it does not matter whether the rescaled 2-norm or the 2-norm is used.  

The errors are from two sources: spatial discretization of the Laplacian  and potential operator, and the time discretization by Trotter formulae. 
We first bound the error from spatial discretization in \cref{lem:harmonic_space_error}, and its proof is given in \cref{append:proof_harmonic_space_error}.
\begin{lemma}\label{lem:harmonic_space_error}
    Let the exact solution of the Schr\"odinger equation with Hamiltonian~\cref{eqn:schrodinger_tdmass}  be $\phi(t,x)$. 
    Then 
    \begin{enumerate}
        \item for any $0\le t\le T,x\in [0,1]$, 
        \begin{equation}
        \begin{split}
            & \quad |n^2(\phi(t,x+1/n)-2\phi(t,x)+\phi(t,x-1/n)) - \Delta \phi(t,x)| \\
            & \leq \frac{1}{3n^2} \sup_{y\in[0,1]} \left|\frac{\partial^4}{\partial x^4}\phi(t,y)\right|. 
        \end{split}
    \end{equation}
        \item Let $\vec{\psi}(t)$ denote the solution of the dynamics
        \begin{equation} \label{eqn:schd_discretized_h1_h2}
            \I \partial_t \vec{\psi}(t) = (f_1(t)H_1 + f_2(t)H_2)\vec{\psi}(t)
        \end{equation}
        where $H_1$ and $H_2$ are the discretized Hamiltonian defined in \cref{eqn:H_1_fd} and \cref{eqn:H_2_fd}, then for any $0\le t\in T$,
        \begin{equation}\label{eqn:schrodinger_tdmass_spatial_error}
            \|(\phi(t,k/n))_{k=0}^{n-1}-\vec{\psi}(t)\|_{\star} \leq \frac{t}{3n^2} \|f_1\|_{\infty}\sup_{s\in[0,t], y\in[0,1]} \left|\frac{\partial^4}{\partial x^4}\phi(s,y)\right|. 
        \end{equation}
    \end{enumerate}
\end{lemma}
This lemma shows that in order to make the vector error induced by the spatial discretization bounded by $\epsilon$, it suffices to choose 
    \begin{equation} \label{eqn:n_estimate}
        n = \Or \left( \frac{T^{1/2}}{\epsilon^{1/2}} \left(\sup \left|\frac{\partial^4 \phi}{\partial x^4}\right|\right)^{1/2}\right).
    \end{equation}
The second source of the error is the time discretization using standard or generalized Trotter formula by applying the aforementioned Theorems for the errors in the operator and vector norm. The following lemma makes explicit the scaling in $n$.
\begin{lemma}\label{lem:harmonic_time_error}
    Let $\vec{\psi}(t)$ be the solution of spatially discretized Schr\"odinger equation \cref{eqn:schrodinger_tdmass} using finite difference with $n$ grid points, and $U(t,0)$ be the evolution operator.  Let $\vec{\psi}_{s,p}(t)$ and $\vec{\psi}_{g,p}(t)$ be the corresponding numerical solution from $p$-th order standard and generalized Trotter formula with $L$ equidistant time steps, respectively, and $U_{s,p}(t,0)$, $U_{g,p}(t,0)$ be the corresponding evolution operators. 
    Assume that $L$ is sufficiently large, then there exists a constant $\widetilde{C} > 0$, independent of $T,L,H_1,n,\vec{\psi},\phi$, such that 
    \begin{enumerate}
        \item 
        \begin{equation}
            \begin{split}
                \|U_{s,1}(T,0)-U(T,0)\| &\leq \widetilde{C}\frac{n^2 T^2}{L} ,\\
                \|U_{g,1}(T,0)-U(T,0)\| &\leq \widetilde{C} \frac{n T^2}{L} ,\\
                \|U_{s,2}(T,0)-U(T,0)\| &\leq \widetilde{C} \frac{n^2 T^3}{L^2} ,\\
                \|U_{g,2}(T,0)-U(T,0)\| &\leq \widetilde{C} \frac{n^2 T^3}{L^2} .
            \end{split}
        \end{equation}
        \item 
         \begin{equation}
            \begin{split}
                \|\vec{\psi}_{s,1}(T)-\vec{\psi}(T)\|_{\star} &\leq \widetilde{C} \frac{T^2}{L}  \left((T+1)\sup\left|\frac{\partial^4 \phi}{\partial x^4}\right| + \sup \left|\frac{\partial^2 \phi}{\partial x^2}\right| + \sup_{y\in[0,1]} \left|\phi(0,y)\right|\right) ,\\
                \|\vec{\psi}_{g,1}(T)-\vec{\psi}(T)\|_{\star} &\leq \widetilde{C} \frac{T^2}{L}  \left[\left((T+1)\sup\left|\frac{\partial^4 \phi}{\partial x^4}\right| + \sup \left|\frac{\partial^2\phi}{\partial x^2}\right|\right)^{1/2}\left(\sup_{y\in[0,1]} \left|\phi(0,y)\right|\right)^{1/2} + \sup_{y\in[0,1]} \left|\phi(0,y)\right|\right] ,\\
                \|\vec{\psi}_{s,2}(T)-\vec{\psi}(T)\|_{\star} &\leq \widetilde{C} \frac{T^3}{L^2}  \left((T+1)\sup\left|\frac{\partial^4 \phi}{\partial x^4}\right| + \sup \left|\frac{\partial^2 \phi}{\partial x^2}\right| + \sup_{y\in[0,1]} \left|\phi(0,y)\right|\right) ,\\
                \|\vec{\psi}_{g,2}(T)-\vec{\psi}(T)\|_{\star} &\leq \widetilde{C} \frac{T^3}{L^2}  \left((T+1)\sup\left|\frac{\partial^4 \phi}{\partial x^4}\right| + \sup \left|\frac{\partial^2 \phi}{\partial x^2}\right| + \sup_{y\in[0,1]} \left|\phi(0,y)\right|\right) ,
            \end{split}
        \end{equation}
        where the notation $\sup$ without any subscript should be interpreted as $\sup_{t\in[0,T],x\in [0,1]}$.
    \end{enumerate}
\end{lemma}

\begin{rem}
The condition that $L$ should be sufficiently large is to ensure that the lowest order term in the error bounds are dominant and to allow us to discard the higher order terms.  This can be guaranteed by requiring the desired level of error $\epsilon$ to be sufficiently small in the complexity estimate later. 
\end{rem}

\begin{proof}
  1. The result follows by combining \cref{thm:operator_norm_bound} and the scaling of the matrix norms provided in \cref{lem:harmonic_H_commu_scaling}.
    
    \smallskip
    
    2. According to \cref{thm:vector_norm_bound} and the fact that $\|\vec{\psi}(0)\|_{\star} \leq \sup_{y\in[0,1]} |\psi(0,y)|$, we only need to bound $\|H_1\vec{\psi}(t)\|_{\star}$ for any $t\in[0,T]$. 
    Let $r(t,x) = n^2(\phi(t,x+1/n)-2\phi(t,x)+\phi(t,x-1/n))-\Delta\phi(t,x)$ where $\phi(t,x)$ is the exact solution before any discretization. 
    By \cref{lem:harmonic_space_error}, 
    \begin{equation}
    \begin{split}
        \|H_1\vec{\psi}(t)\|_{\star} & \leq \|H_1(\vec{\psi}(t)-(\phi(t,k/n))_{k=0}^{n-1})\|_{\star} + \|H_1(\phi(t,k/n))_{k=0}^{n-1} \|_{\star} \\
        & \leq \|H_1 (\vec{\psi}(t)-(\phi(t,k/n))_{k=0}^{n-1}\|_{\star} + \|( (\Delta\phi(t,k/n))_{k=0}^{n-1}\|_{\star} + \|(r(t,k/n))_{k=0}^{n-1}\|_{\star} \\
        & \leq \widetilde{C} \left(t \sup_{s\in[0,T],y\in[0,1]}\left|\frac{\partial^4}{\partial x^4}\phi(s,y)\right| + \sup_{y\in[0,1]}\left|\frac{\partial^2}{\partial x^2}\phi(t,y)\right| + \frac{1}{n^2}\sup_{s\in[0,T],y\in[0,1]}\left|\frac{\partial^4}{\partial x^4}\phi(s,y)\right|\right) \\
        & \leq \widetilde{C} \left((T+1)\sup_{s\in[0,T],y\in[0,1]}\left|\frac{\partial^4}{\partial x^4}\phi(s,y)\right| + \sup_{s\in [0,T],y\in[0,1]}\left|\frac{\partial^2}{\partial x^2}\phi(s,y)\right|\right)
    \end{split}
    \end{equation}
\end{proof}

Finally, we combine both spatial and temporal errors. It is not possible to obtain an operator norm error bound between the evolution operator of an unbounded operator and that of a finite dimensional matrix. Hence the operator norm bounds below are obtained by plugging in the estimate of $n$ that achieves the vector norm error with precision $\epsilon$. In particular, the operator norm error bound involves the derivatives of the exact solution of interest $\phi$.  Combining \cref{eqn:n_estimate} and \cref{lem:harmonic_time_error}, we obtain the total complexity estimates. 
\begin{theorem}\label{thm:main_harmonic}
    We use central finite difference for spatial discretization and Trotter formulae for time discretization to obtain an $\epsilon$-approximation in rescaled 2-norm of the solution $\phi(t,x)$. 
    Let $L_{\text{ope},s,1}$ and $L_{\text{ope},g,1}$ denote the total number of required time steps of first-order standard and generalized Trotter estimated from operator norm error bounds, respectively, $L_{\text{vec},s,1}$ and $L_{\text{vec},g,1}$ denote the estimates from vector norm error bounds, and $L_{\text{ope},s,2}$, $L_{\text{ope},g,2}$, $L_{\text{vec},s,2}$, $L_{\text{vec},g,2}$ are the corresponding estimates for second-order schemes. Then for sufficiently small $\epsilon$, 
    \begin{equation}
    \begin{split}
        &L_{\text{ope},s,1} = \Or\left(\frac{ T^3}{\epsilon^2} \left(\sup \left|\frac{\partial^4 \phi}{\partial x^4}\right|\right)\right), \\
        &L_{\text{ope},g,1} = \Or\left(\frac{ T^{5/2}}{\epsilon^{3/2}} \left(\sup \left|\frac{\partial^4 \phi}{\partial x^4}\right|\right)^{1/2}\right), \\
        &L_{\text{vec},s,1} = \Or\left(\frac{T^2}{\epsilon}\left((T+1)\sup\left|\frac{\partial^4 \phi}{\partial x^4}\right| + \sup \left|\frac{\partial^2 \phi}{\partial x^2}\right| + \sup_{y\in[0,1]} \left|\phi(0,y)\right|\right) \right) , \\
        &L_{\text{vec},g,1} = \Or\left(\frac{T^2}{\epsilon}\left[\left((T+1)\sup\left|\frac{\partial^4 \phi}{\partial x^4}\right| + \sup \left|\frac{\partial^2\phi}{\partial x^2}\right|\right)^{1/2}\left(\sup_{y\in[0,1]} \left|\phi(0,y)\right|\right)^{1/2} + \sup_{y\in[0,1]} \left|\phi(0,y)\right|\right]\right) , 
    \end{split}
    \end{equation}
    and 
    \begin{equation}
    \begin{split}
        &L_{\text{ope},s,2} = L_{\text{ope},g,2} = \Or\left(\frac{ T^{2}}{\epsilon} \left(\sup \left|\frac{\partial^4 \phi}{\partial x^4}\right|\right)^{1/2}\right), \\
        &L_{\text{vec},s,2} =  L_{\text{vec},g,2} = \Or\left(\frac{T^{3/2}}{\epsilon^{1/2}} \left((T+1)\sup\left|\frac{\partial^4 \phi}{\partial x^4}\right| + \sup \left|\frac{\partial^2 \phi}{\partial x^2}\right| + \sup_{y\in[0,1]} \left|\phi(0,y)\right|\right)^{1/2} \right). 
    \end{split}
    \end{equation}
\end{theorem}
\begin{proof}
    The complexity can be estimated by requiring both error bounds in \cref{lem:harmonic_space_error} and \cref{lem:harmonic_time_error} to be smaller than $\epsilon$. 
    First, by requiring the right hand side of \cref{eqn:schrodinger_tdmass_spatial_error} to be bounded by $\epsilon$, the scaling of $n$ should be that in \cref{eqn:n_estimate}. 
    Plug this back into \cref{lem:harmonic_time_error} and also let the bounds in \cref{lem:harmonic_time_error} to be bounded by $\epsilon$, we obtain the complexity estimates. 
\end{proof}

\cref{thm:main_harmonic} clearly illustrates the advantage of vector norm error bounds in terms of the desired level of error $\epsilon$. 
More precisely, the total number of required Trotter steps estimated from vector norm bounds only scales $\Or(1/\epsilon^{1/p})$ for $p$-th order schemes.
This is  because the operator norm error bounds depend on the spatial discretization $n$, where $n = \Or(1/\epsilon^{1/2})$ for second order spatial discretization, but the vector norm error bounds do not. We summarize our complexity estimates in terms of the spatial discretization as well as the error level $\epsilon$ in \cref{tab:trotter_summary}, where the simulation time $T$ is $\Or(1)$. 

The best scaling is achieved by the second order standard and generalized Trotter formulae with the vector norm error bound, which is the result we are referring to as `This work' in \cref{tab:trotter_comparison} for comparison with existing estimates. As discussed earlier, in order to demonstrate the behavior of the Trotter formulae for unbounded operators, we require $n$ to grow as $\operatorname{poly}(1/\epsilon)$. Therefore we choose $V(x)$ to be a $C^4$ function so that the commutator scaling of the second order Trotter formulae are valid.

Numerical tests in \cref{sec:numer} demonstrate that the complexity estimates from vector norm error bounds are sharp in terms of $\epsilon$ for all the schemes we consider. 

\begin{rem}[\REV{\emph{a priori} estimates of the solution $\phi$}]
\REV{Due to the potential growth of the derivatives of the exact solution with respect to $T$, \emph{a priori} estimates are required if we would like to obtain the overall scalings in $T$.} Such \emph{a priori} estimates, where the spatial derivatives are bounded by polynomials of $T$, have been established in the literature for various special cases, such as when $f_1 \equiv 1$, $f_2$ is smooth in $t$ and $V$ is a real potential, smooth in $x$ and periodic in $x$ as considered in \cite{Bourgain1999}, and for strictly positive $f_1$ in \cite{Montalto2018}. \REV{The corresponding estimates are usually  technical, while the common approach to derive them is a combination of various analytical tools and a careful capture on the resonance in the dynamics.} \REV{Detailed discussions are orthogonal to the topic here and are omitted. }
\end{rem}

As we have already observed in \cref{thm:operator_norm_bound}, the generalized Trotter formula exhibits commutator type error bounds, while the standard Trotter formula does not. However, the commutator error bound only translates to improved asymptotic complexity for the first order generalized Trotter scheme.  
For second order schemes, there is no significant difference in the scaling with respect to $\epsilon$  between the standard and generalized Trotter formulae. As discussed before, this is due to the fact that $\norm{H_1}$ and $\norm{[H_1,[H_1,H_2]]}$ have the same asymptotic scaling in $n$. The generalized Trotter is less restrictive on the control functions, namely, the $p$-th order generalized Trotter formula ($p = 1,2$) only requires the boundedness of the derivatives of control functions up to the $(p-1)$-th order while the $p$-th standard one requires the boundedness up to the $p$-th order. We expect the same situation for higher order Trotter formulae. 

\section{Numerical Results}\label{sec:numer}
To illustrate the difference between the operator norm and vector norm, we consider the following Hamiltonian
\begin{equation}
H(t) = -\frac{1}{2}(2 + \sin(at+0.5)) \Delta + (1 + \cos(t)) V(x), \quad V(x) = 1 - \cos(x), \quad x\in [-\pi,\pi]
\label{eqn:}
\end{equation}
with periodic boundary conditions. Here $a>0$ controls the magnitude of the derivatives $\|f_1'\|_\infty$ and $\|f_1''\|_\infty$. These sizes play a role in the preconstants of the errors as shown in \cref{thm:error_trotterstep}.  As discussed in \cref{sec:harmonic}, $H_1$ and $H_2$ correspond to the discretized matrices of $-\Delta$ and $V(x)$, respectively. Besides the second order finite difference scheme, we also demonstrate that our estimates are equally applicable to the Fourier discretization. Though in this particular example $V(x)$ is smooth, the scaling of $n$ is still chosen according to \cref{eqn:n_estimate}, which only requires the regularity of $V$ up to its fourth order derivatives and hence works for more general potentials.

We first demonstrate the scaling of the vector norms and the operator norms, respectively.  Consider the vector $\vec{v}$ as the discretization of the smooth function $\cos(x)$. \cref{fig: op_vec_norms} plots the operator norms and the vector norms for various number of spatial grids $n$ using the  finite difference and Fourier spatial discretization. We find that $\|[H_1, [H_1, H_2]]\|$  grows quadratically with respect the the number of spatial grids while $\|[H_1, H_2] \|$ scales linearly, which agrees with \cref{lem:harmonic_H_commu_scaling}.  However, the vector norms $\|[H_1, [H_1, H_2]]\vec{v}\|_\star$, $\|[H_1, H_2] \vec{v}\|_\star$ remain of the same scale. This behavior is not restricted to the specific spatial discretization. Moreover, the scalings of $\|D_1 \vec{v} \|_\star$ and $\|H_1 \vec{v}\|_\star$ are found to be the same as those of $\|[H_1, [H_1, H_2]]\vec{v}\|_\star$ and $\|[H_1, H_2]\vec{v}\|_\star$. This verifies the in assumptions \cref{eqn:assump_comm1,eqn:assump_comm2}, which is also proved for the finite difference scheme in \cref{lem:harmonic_verification_assumption}.

\begin{figure}
    \centering
    \subfloat[Finite difference discretization]{
        \includegraphics[width=.49\textwidth]{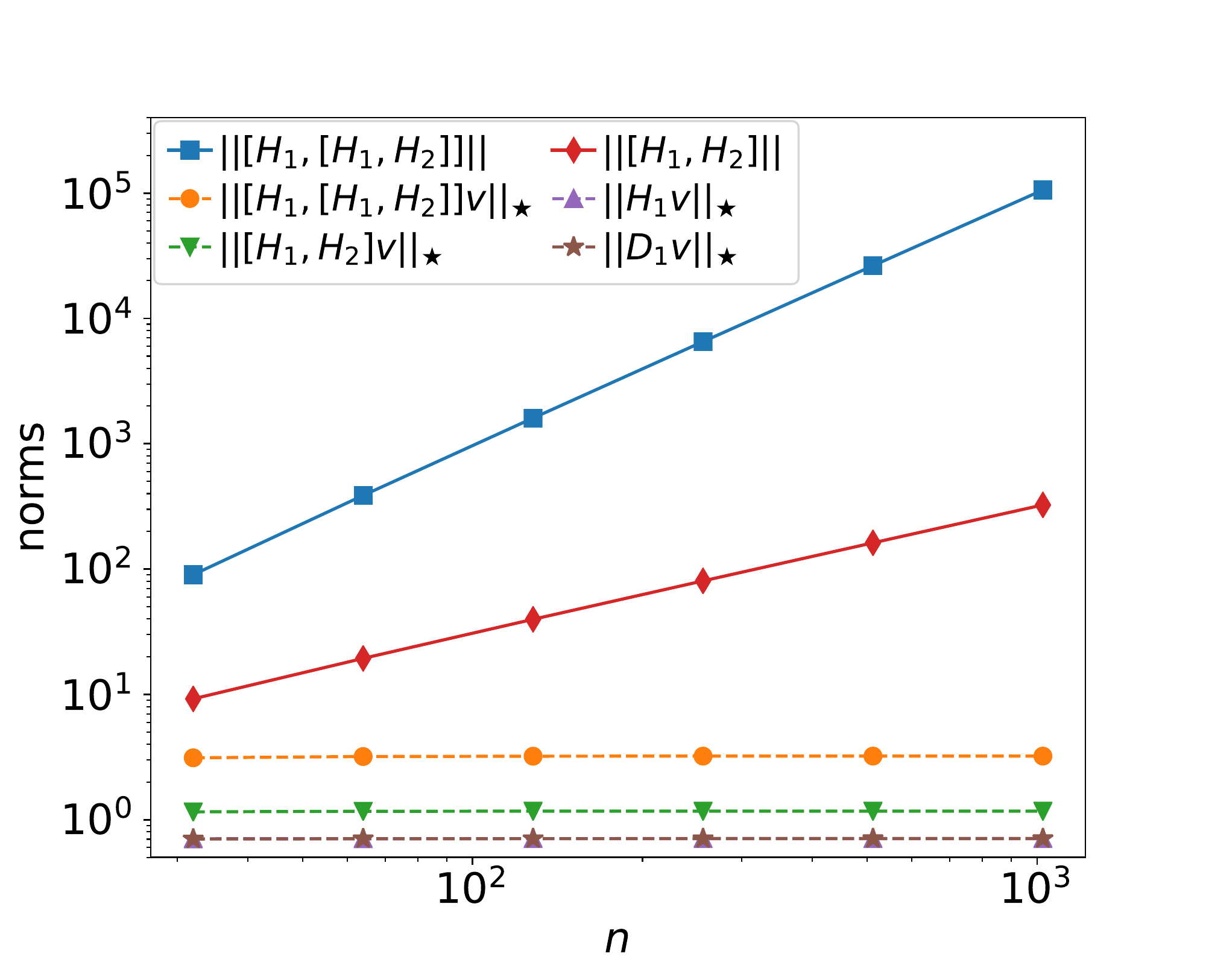}}
    \subfloat[Fourier discretization]{
        \includegraphics[width=.49\textwidth]{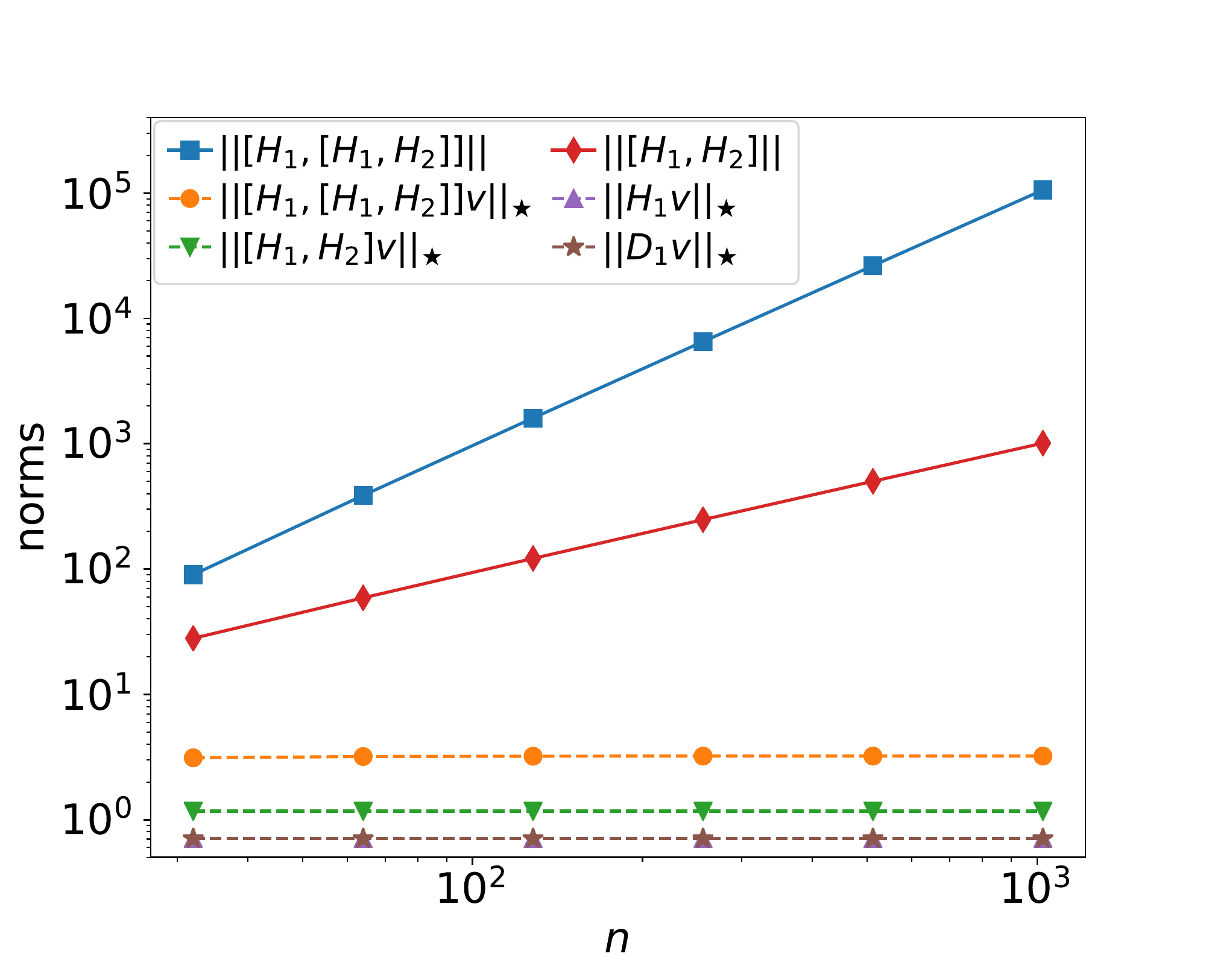}}
    \caption{Operator and vector norms of a smooth vector for various numbers of spatial grids $n$. $\|[H_1, [H_1, H_2]]\|$ and $\|[H_1, H_2]\|$ scales quadratically and linearly with respect to $n$, but the vector norms do not grow as $n$ gets larger.}
    \label{fig: op_vec_norms}
\end{figure}

We then verify the scaling of the errors with respect to $n$. The initial wavefunction is $\phi(x,0)=\cos(x)$. The time step size $h$ is fixed to be $10^{-4}$. We run the Trotter formulae for 10 steps, which is sufficient for demonstrating the difference in scalings.  The relative errors for both the operator and vector norms are plotted in \cref{fig: num_errors} for $a = 1$ and $a = 10$.  In terms of the operator norm, the generalized Trotter formula has a smaller error compared to the standard one: the relative error in the operator norms for the first-order standard Trotter scheme scales quadratically with respect to the number of grids while the first-order generalized Trotter schemes admits a linear scaling thanks to the commutator bounds. On the other hand, the relative errors in the vector norm do not grow with respect to $n$. 

For second-order schemes, it can been seen that the errors measured by the operator norm for both methods grow quadratically with respect to $n$, while the corresponding errors in the vector norm are stable as $n$ increases. These results agree with ~\cref{lem:harmonic_time_error}. Note that though the operator norm errors of the second-order schemes have the same asymptotic scaling in $n$, their preconstants may differ. When $a = 1$, the sizes of $\norm{f_1}_{\infty},\norm{f_1'}_{\infty},\norm{f_1''}_{\infty}$ are comparable, and there is no significant difference in the preconstants. However, when $a =10$,  $\norm{f_1''}_{\infty}$ is one order of magnitude larger than $\norm{f_1'}_{\infty},\norm{f_1}_{\infty}$. In this case, we find from \cref{fig: num_errors} that the generalized Trotter formula has a smaller preconstant, which agrees with the preconstant estimates as described in \cref{thm:error_trotterstep}.

Moreover, we compare the scaling of the number of Trotter steps for various precision $\epsilon$, measuring the relative error via the vector norm. We fix $a = 10$, $T = 0.16$, and consider the precision $\epsilon$ as $2^{-10}$, $2^{-12}$, $2^{-14}$, $2^{-16}$, $2^{-18}$, $2^{-20}$ and take $n \propto \epsilon^{-0.5}$ as $2^5$, $2^6$, $2^7$, $2^8$, $2^9$ and $2^{10}$. As is presented in \cref{fig:vec_error_eps_scaling}, both second-order Trotter formulae requires the number of Trotter steps $L = \Or(\epsilon^{-0.5})$ while it requires $L = \Or(\epsilon^{-1})$ for both first-order Trotter formulae.  These results agree with \cref{thm:main_harmonic}.


\begin{figure}
    \centering
        \subfloat[$a = 1$. Finite difference discretization]{
        \includegraphics[width=.48\textwidth]{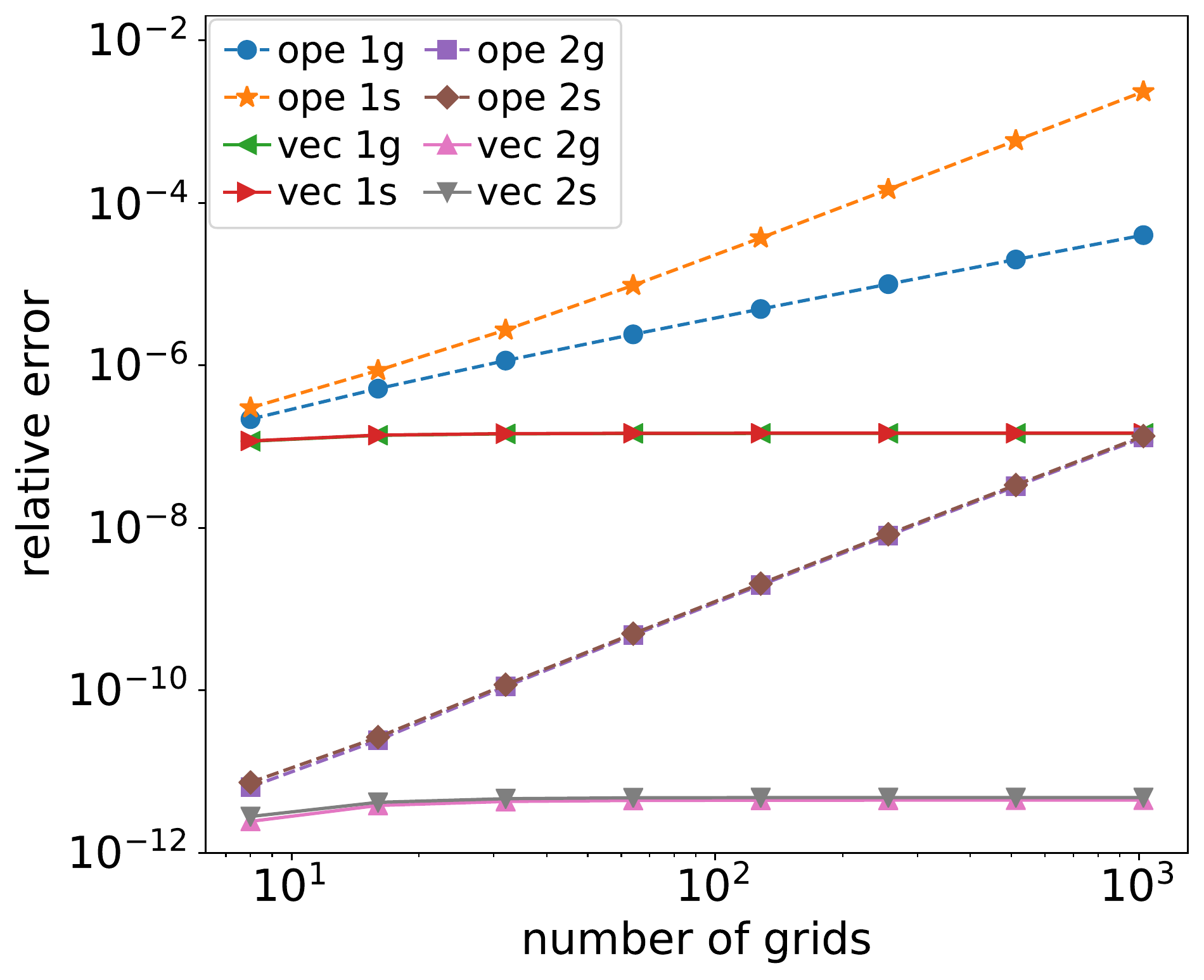}}
    \subfloat[$a = 1$. Fourier discretization]{
        \includegraphics[width=.49\textwidth]{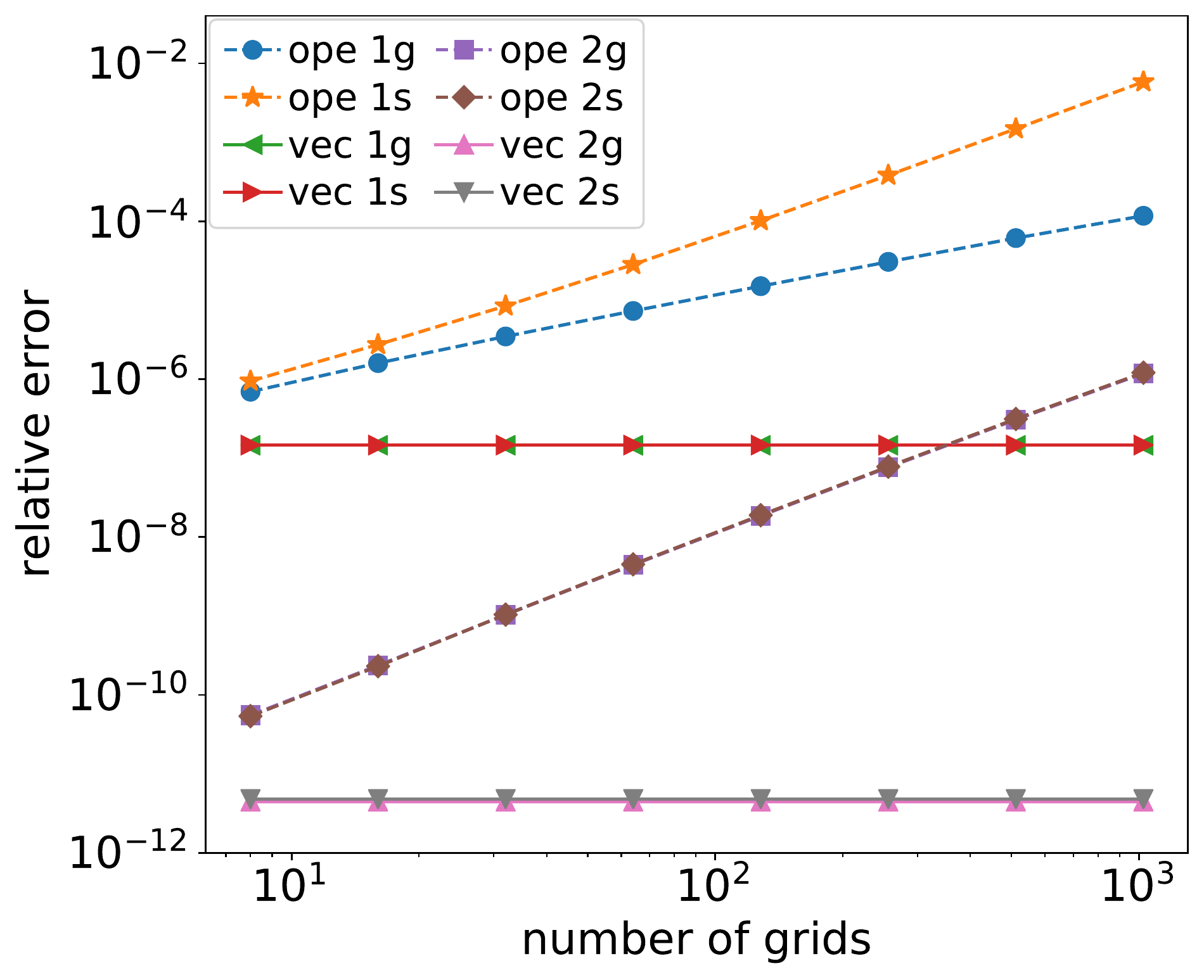}}\\
    \subfloat[$a = 10$. Finite difference discretization]{
        \includegraphics[width=.48\textwidth]{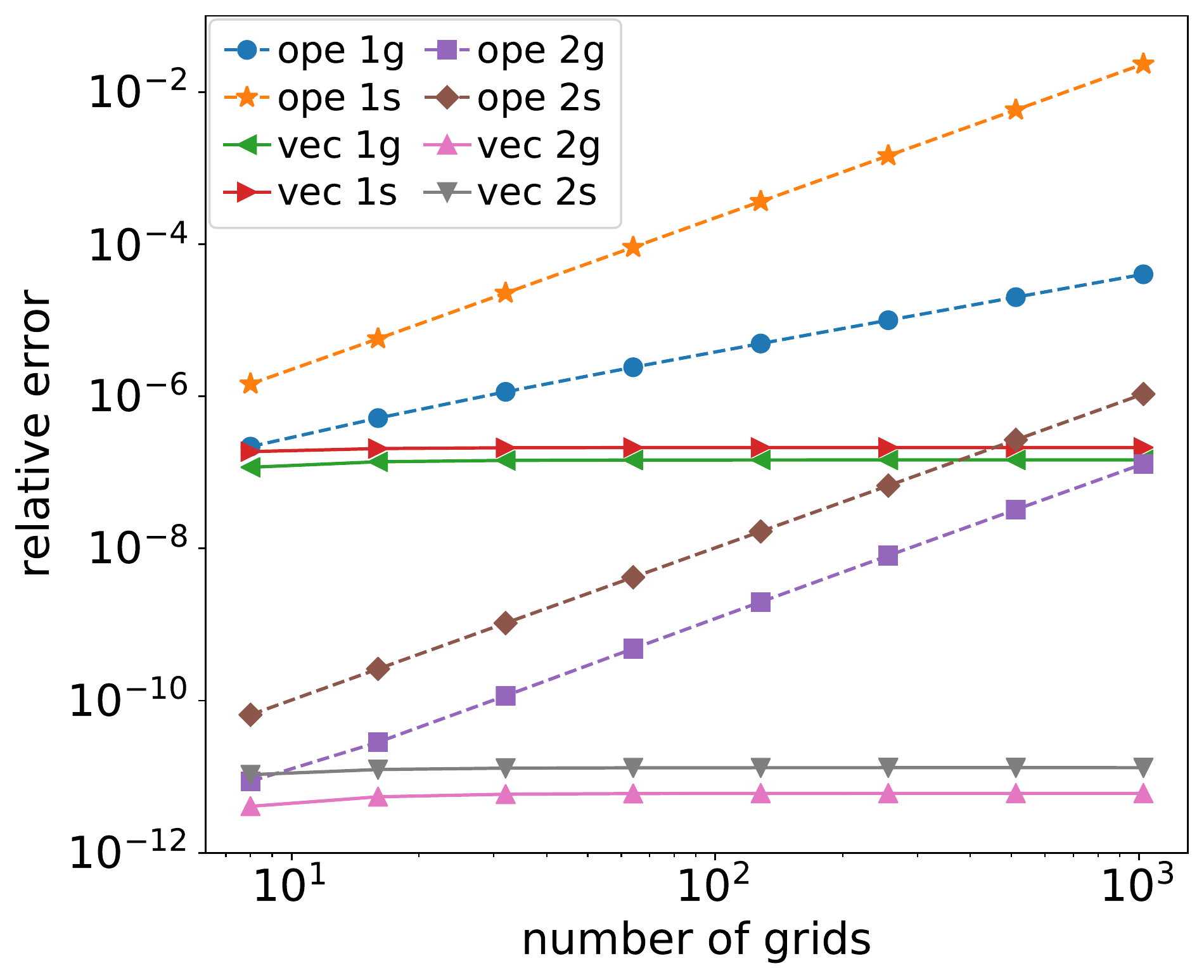}}
    \subfloat[$a = 10$. Fourier discretization]{
        \includegraphics[width=.49\textwidth]{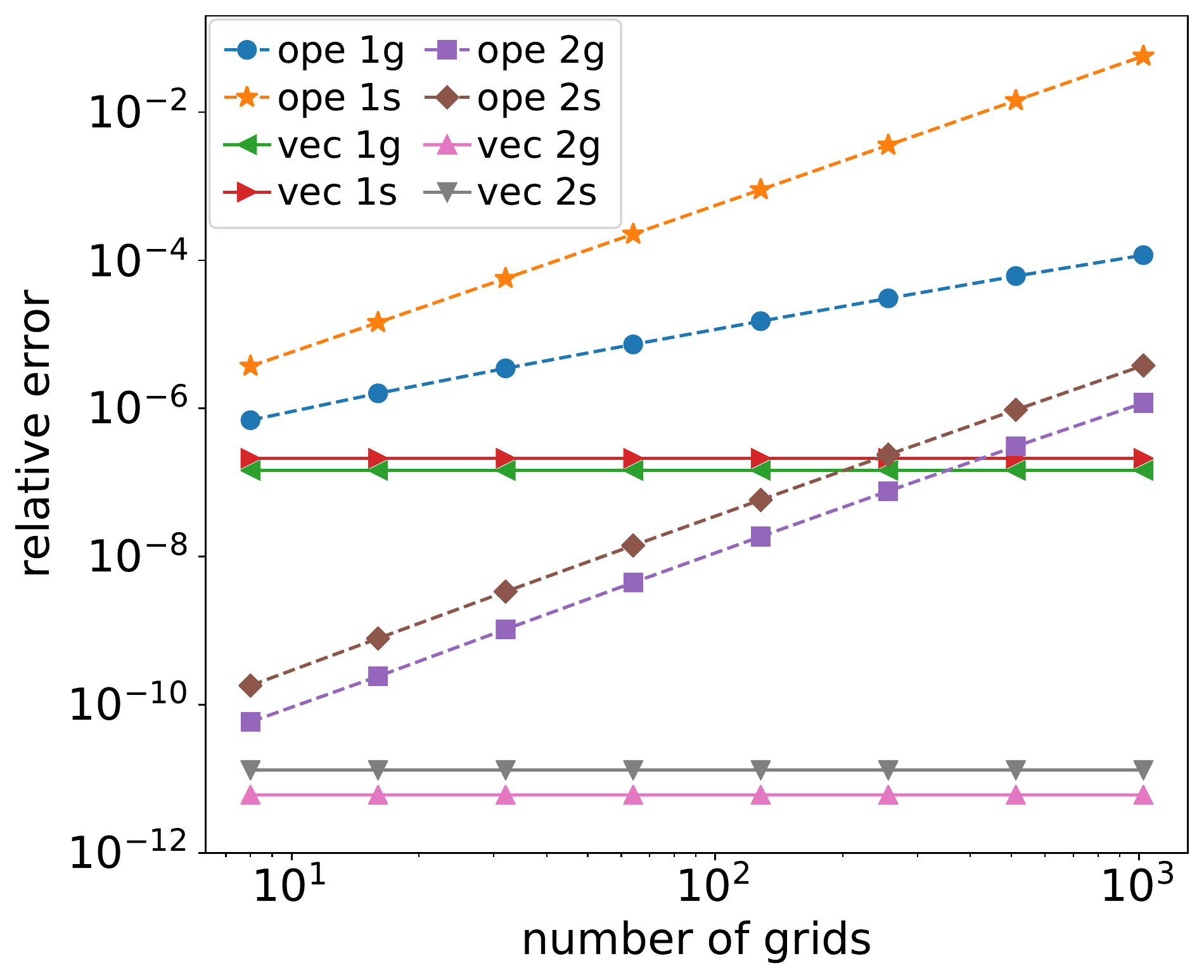}} 
    \caption{Relative Errors in the operator and vector norms. In the legend, ``g'' stands for the generalized Trotter formula and ``s'' for the standard Trotter formula. The error in operator norm is labeled as ``ope'' while the one in vector norm as ``vec''. First Row: $a = 1$ with slowly varying control functions. Second Row: $a = 10$ with fast varying control functions.}
    \label{fig: num_errors}
\end{figure}

\begin{figure}
    \centering
    \includegraphics[width=.55\textwidth]{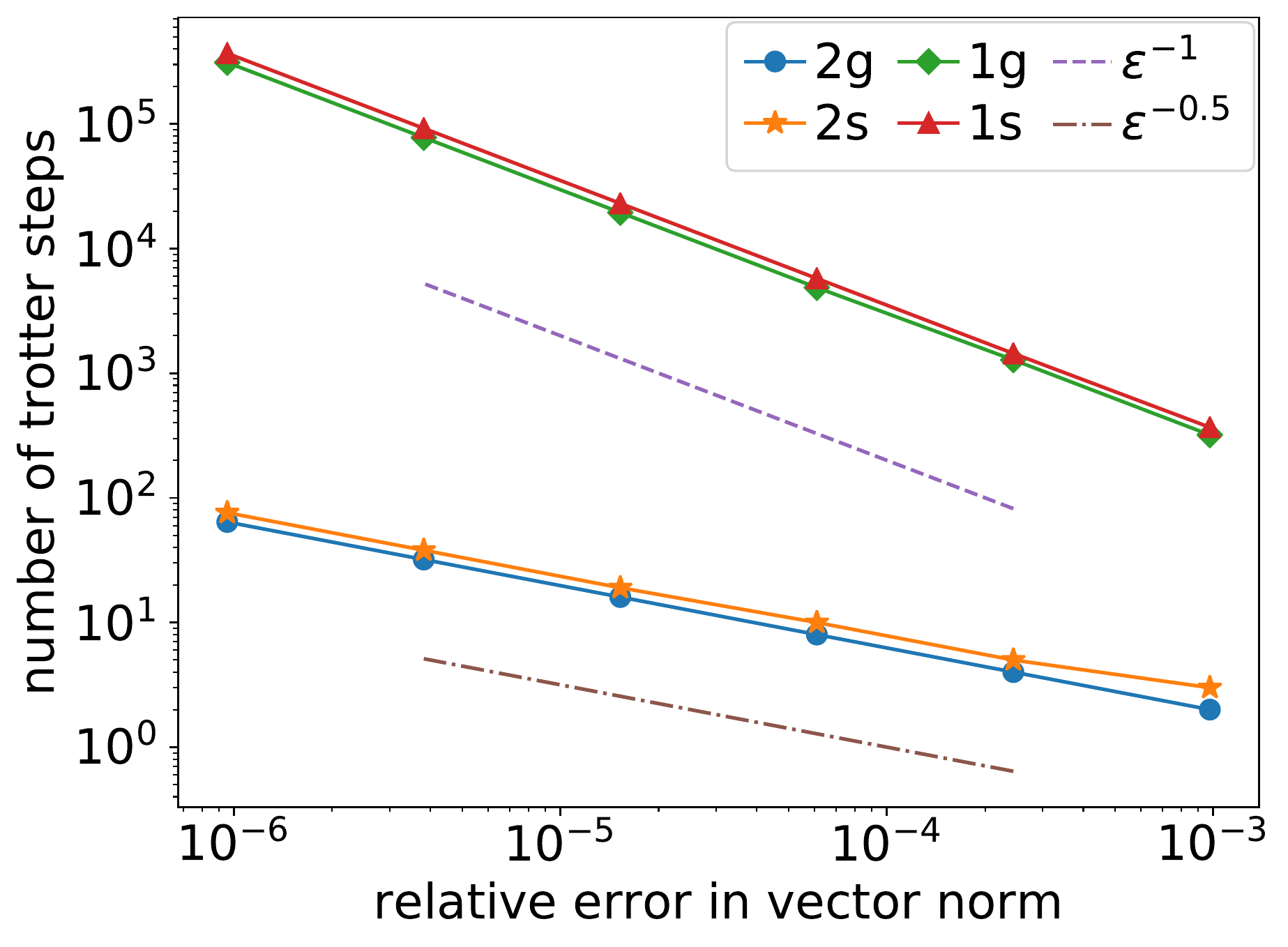}
    \caption{The number of Trotter steps required to achieve various precision for the relative error in the vector norm. The spatial discretization is finite difference. Both second-order Trotter formulae scales proportionally to $\epsilon^{-0.5}$ while the first-order formulae scales as $\epsilon^{-1}$, which agrees with the theoretical bounds.}
    \label{fig:vec_error_eps_scaling}
\end{figure}

\section{Conclusion}\label{sec:conclusion}

We have studied in detail the behavior of first and second order standard and generalized Trotter formulae for time-dependent Hamiltonian simulation with unbounded, control type Hamiltonians. We demonstrated that the error of the Hamiltonian simulation for a given initial state can often be overestimated using the standard analysis based on operator norms, which overestimates the computational cost. By taking into account the information of the initial state in the error analysis, sharper error estimates can be derived via the vector norm scaling. As a side product, we also obtained improved error bounds of the standard and generalized Trotter formulae in operator norm as well in the time-dependent setting.

As an example, we applied our results to the time-dependent Schr\"odinger equation with a time-dependent effective mass and frequency. While the complexities of existing quantum algorithms for time-dependent Hamiltonian simulation scale at least linearly in the spatial discretization parameter $n$, we demonstrate that, the error bounds in vector norm do not suffer from such overheads (for both the standard and generalized Trotter formulae). Thus in this setting, our results outperform all existing quantum algorithms, including higher order Trotter and post-Trotter methods. 

The bilinear form in \cref{eqn:control_ham} facilitates the discussion of the error of the Trotter formulae. 
For the most general Hamiltonian $H(t) = H_1(t)+H_2(t)$, it has been demonstrated that the error bound can be much more complicated even in the second order case~\cite{HuyghebaertDeRaedt1990}.
Nevertheless, under suitable modifications, the main conclusion of this paper can still be applicable to more general time-dependent Hamiltonian under further assumption that $\partial_t^{k} H_j(s)$ and $\partial_t^{k'} H_j(s')$ commute for any  $j = 1, 2$ and $k,k',s,s'$ (thus no essential difference is introduced in taking derivatives of unitaries and deriving error representation).
This allows us to simulate e.g. Schr\"odinger equation with general time-dependent potential function $V(x,t)$.

A natural extension of this work is to consider general high order time-dependent standard and generalized Trotter formulae defined by Suzuki recursion~\cite{Suzuki1993,WiebeBerryHoyerEtAl2010}.
For the operator norm error bound, our results can be generalized to higher order case with a control Hamiltonian \cref{eqn:control_ham}. 
More specifically, let $\mathcal{C}_{k}$ denote the set of the norms of all possible $k$-th order nested commutators of $H_1$ and $H_2$, for example $\mathcal{C}_0 = \{\|H_1\|,\|H_2\|\}$, $\mathcal{C}_1 = \{\|[H_1,H_2]\|\}$, and $\mathcal{C}_2 = \{\|[H_1,[H_1,H_2]]\|, \|[H_2,[H_2,H_1]]\|\}$. 
For $p$-th order schemes, we expect that the one-step operator norm error bounds for the standard and generalized Trotter formula scales as $\Or(\alpha_{s,p}h^{p+1})$, $\Or(\alpha_{g,p}h^{p+1})$, respectively. Here $\alpha_{s,p}$ is a linear combination of terms in the set $\bigcup_{k=0}^{p} \mathcal{C}_k$, while $\alpha_{g,p}$ is expressed as a linear combination of terms in the set $\bigcup_{k=1}^{p} \mathcal{C}_k$. Hence the difference lies in whether $\mathcal{C}_0$ is included, and generalized Trotter formula allows a commutator scaling. Notice that such an error bound will improve the best existing estimate~\cite{WiebeBerryHoyerEtAl2010}\REV{, which depends on the norms of the Hamiltonians as well as their high order derivatives, and does not demonstrate possible commutator scalings}.  



The extension of our vector norm error bounds to $p$-th order time-dependent Trotter formula is also possible. 
The corresponding assumption on the bounds of commutators (\emph{i.e.} counterpart of \cref{assump:commutator_bound} in this work) becomes 
\begin{equation}\label{eqn:assumption_high_order}
    \|[\underbrace{H_1,[H_1,\cdots,[H_1}_{k\text{ repeats}},H_2]]\cdots ]\vec{v}\| \leq \Or\left(\|H_1^{k/2}\vec{v}\|+\|\vec{v}\|\right) 
\end{equation}
for any $1 \leq k \leq p$. 
Compared with the operator norm error bounds, for the Schr\"odinger equation with a time-dependent mass and frequency, such vector norm error bounds can still remove the dependence on the spatial discretization thus provide speedup in terms of the accuracy. However, the significance of such improvement might be subtle: in order to satisfy the assumption in \cref{eqn:assumption_high_order}, the potential function $V(x)$ needs to be much smoother with bounded higher order derivatives. Hence, the dependence of $n$ on the error $\epsilon$ may become much weaker by employing higher order discretization schemes. In such a scenario, the spectral norms $\norm{H_1}$ and $\norm{H_2}$ may even become comparable, and the Hamiltonian $H(t)$ may not be regarded as an unbounded operator after all.

 In this work, we mainly focus on the improvement brought by vector norm error bounds in terms of the accuracy. 
It is also interesting to study whether vector norm error bounds can improve the scalings of other parameters. 
For example, if the Schr\"odinger equation is in $d$ dimension rather than one dimension considered in this paper, then a vector norm error bound may offer speedup in terms of $d$, since the degree of freedom for spatial discretization can scale linearly in $d$~\cite{KivlichanWiebeBabbushEtAl2017}. 
Another related topic is the scaling with respect to the number of the particles in quantum many-body systems.  Recently 
\cite{SuHuangCampbell2020} obtained an improved estimate in terms of the number of electrons for electronic structure problem with  plane-wave basis functions in a second quantized formulation, by combining sparsity, commutator scalings and initial-state knowledge and bounding the operator norm on an $\eta$-electron sub-manifold. 
Although much smaller than that on the entire space, the operator norm on the $\eta$-electron sub-manifold may still overestimate the error, and a vector norm error bound might offer further improvement by taking the smoothness and low-energy property of the wavefunction into consideration. 
It is also an interesting question to investigate whether a vector norm error bound can provide any benefit for other applications such as spin systems.

Our work suggests that it may be of interest to explore the gap between the operator norm and vector norm error bounds in other schemes for Hamiltonian simulations with unbounded operators. Note that such a gap may not exist in all methods. For instance, for time-independent Hamiltonian simulation, the quantum signal processing (QSP) method~\cite{LowChuang2017} is based on the polynomial approximation to the function $\cos(xt)$ and $\sin(xt)$, and we do not expect that the error bound can be significantly improved by considering vector norms. However, it may be possible to prove vector norm error bounds for other post-Trotter methods.


\section*{Acknowledgments:} 

This work was partially supported by the Department of Energy under Grant No.  DE-SC0017867 (D.A.,L.L.), under Grant No. DE-AC02-05CH11231, No. FWP-NQISCCAWL, and by the National Science Foundation under the QLCI program through grant number OMA-2016245 (L.L.). We thank Andrew Childs, Jinpeng Liu, Yuan Su, Yu Tong for helpful discussions.

\bibliographystyle{abbrvurl}
\bibliography{gtrotter}

\clearpage
\appendix

\appendix

\section{Derivation of results in \cref{tab:trotter_comparison}}\label{app:derivation_table}

In this section we show explicitly how to derive the results in \cref{tab:trotter_comparison}. 
Throughout this section we are considering the setup in \cref{sec:harmonic} with $T = \Or(1)$. 

To obtain \cref{tab:trotter_comparison}, we first restate all the complexity estimates for different methods proved in existing literature and show how they depend on $\epsilon$ as well as the scale of the Hamiltonians $H_1$ and $H_2$. 
The dependence on $H_1$ naturally gives rise to the dependence on $n$, by noticing that
$$\|H_1\| = \Or(n^2), \quad \|H_2\| = \Or(1), \quad \|[H_1,H_2]\| = \Or(n),$$
$$\|[H_1,[H_1,H_2]]\| = \Or(n^2), \quad \|[H_2,[H_2,H_1]]\| = \Or(n),$$
as is discussed in \cref{lem:harmonic_H_commu_scaling}. 
Then, under second order finite difference spatial discretization, \cref{lem:harmonic_space_error} and \cref{eqn:n_estimate} tell that $n$ should be chosen as large as $\Or(\epsilon^{-1/2})$. 
Plugging this back into the complexity estimates leads to the overall scaling in terms of $\epsilon$, as shown in the last column of \cref{tab:trotter_comparison}. 

\subsection{Time-independent schemes}

\paragraph{Time-independent second order Trotter formula} \cite[Proposition 16]{ChildsSuTranEtAl2020} gives an operator norm error bound for time-independent second order Trotter formula that the one-step local Trotter error is bounded by 
$$\frac{h^3}{12} \|[H_2,[H_2,H_1]]\| + \frac{h^3}{24} \|[H_1,[H_1,H_2]]\|,$$
thus the global Trotter error is bounded by 
$$\left(\frac{1}{12} \|[H_2,[H_2,H_1]]\| + \frac{1}{24} \|[H_1,[H_1,H_2]]\|\right) \frac{T^3}{L^2} = \Or\left(\frac{n^2 }{L^2}\right).$$
To bound this by $\epsilon$, it suffices to choose 
$$L = \Or\left(\frac{n}{\epsilon^{1/2}}\right) = \Or\left(\frac{1}{\epsilon}\right). $$

\cite[Theorem 3.2]{JahnkeLubich2000} provides a vector norm error bound for time-independent second order Trotter formula that the global Trotter error is bounded by 
$$\Or\left(h^2 \left(\|H_1\vec{v}\|_{\star} + \|D_1\vec{v}\|_{\star} + \|\vec{v}\|_{\star}\right)\right) = \Or\left(\frac{1}{L^2}\left(\|H_1\vec{v}\|_{\star} + \|\vec{v}\|_{\star}\right)\right).$$
We remark that~\cite{JahnkeLubich2000} does not track explicitly the dependence on $T$. 
Noticing that $\|H_1\vec{v}\|_{\star}$ and $\|\vec{v}\|_{\star}$ are independent of $\epsilon$ and $n$ (shown in the proof of \cref{lem:harmonic_time_error}), the number of required Trotter steps scales 
$$L = \Or\left(\frac{1}{\epsilon^{1/2}}\right).$$

\paragraph{Time-independent high order Trotter formula} \cite[Corollary 12]{ChildsSuTranEtAl2020} shows that for a $p$-th order time-independent Trotter formula, the number of required Trotter steps to obtain an $\epsilon$-approximation of the exact evolution operator is 
$$L = \Or\left(\frac{\widetilde{\alpha}_{\text{comm}}^{1/p}T^{1+1/p}}{\epsilon^{1/p}}\right),$$
where 
$$\widetilde{\alpha}_{\text{comm}} = \sum_{\gamma_1,\cdots,\gamma_{p+1} = 1}^2 \|[H_{\gamma_{p+1}},\cdots[H_{\gamma_2},H_{\gamma_1}]]\|.$$ 
Straightforward bounds for these $p$-th nested commutators are that $$\|[H_{\gamma_{p+1}},\cdots[H_{\gamma_2},H_{\gamma_1}]]\| = \Or(\|H_1\|^{p-2} \|[H_1,[H_1,H_2]]\|) = \Or(n^{2p-2}),$$ 
which results in 
$$L = \Or\left(\frac{n^{2-2/p}}{\epsilon^{1/p}}\right) = \Or\left(\frac{1}{\epsilon}\right).$$ 
Notice that the scaling of $\epsilon$ is not improved by higher order Trotter formula. 
This is because such an estimate is made under the assumption that the potential $V(x)$ is a $C^4$ function, therefore we only have better scaling for nested commutator up to second order. 
If the potential $V(x)$ has higher regularity, we expect better bounds to exist for general nested commutators, just like the case of $[H_1,H_2]$ and $[H_1,[H_1,H_2]]$. 
In particular, although we do not present complete proof in this paper, a continuous analog as well as discretization under Fourier basis suggests that the norm of $p$-th order nested commutator $\|[H_1,\cdots,[H_1,H_2]]\|$ is bounded by $\Or(\|D_1^p\|)$ if $V(x)$ is $(2p)$-th order continuously differentiable. 
In that case the complexity can be improved to $L = \Or(n/\epsilon^{1/p})$, although there is still a linear dependence on $n$.

\paragraph{Truncated Taylor series} \cite[Theorem 1]{BerryChildsCleveEtAl2014} shows that to obtain an $\epsilon$-approximation of the exact evolution operator using truncated Taylor series, the query complexity is 
$$\Or\left(d^2\|H\|_{\max}\frac{\log(d \|H\|_{\max}/\epsilon)}{\log\log(d \|H\|_{\max}/\epsilon)} \right). $$
Here $d$ is the sparsity of the Hamiltonian, $\|H\|_{\max}$ denotes the largest matrix element of $H$ in absolute value. 
Notice that $\|H_1\|_{\max} = \Or(n^2)$ since every non-zero entry of $H_1$ is either $n^2$ or $(-2n^2)$, and $\|H_2\|_{\max} = \Or(1)$, we have $\|H\|_{\max} = \Or(n^2)$. 
Therefore the query complexity becomes 
$$\Or\left(n^2\frac{\log(n^2/\epsilon)}{\log\log(n^2/\epsilon)} \right) = \widetilde{\Or}\left(n^2\right) = \widetilde{\Or}\left(\frac{1}{\epsilon}\right).$$ 
We remark that the work~\cite{KivlichanWiebeBabbushEtAl2017} studies further the complexity of simulating time-independent many-body Hamiltonian and discusses carefully the errors from both time and space discretization. 
In this work, the authors use truncated Taylor series as well for time discretization, and use high order finite difference formula for spatial discretization. 
However, they only assume that the potential $V(x)$ is first-order continuous differentiable thus the high order finite difference formula does not offer improved scaling of $n$ than $\Or(1/\epsilon)$~\cite[Theorem 4]{KivlichanWiebeBabbushEtAl2017}, which results in a total complexity $\widetilde{\Or}(1/\epsilon^2)$~\cite[Theorem 3 \& 4]{KivlichanWiebeBabbushEtAl2017}. 
The scaling can be improved if $V(x)$ becomes smoother.

\paragraph{Quantum signal processing} \cite[Theorem 3]{LowChuang2017} proposes a quantum signal processing approach for time-independent Hamiltonian simulation with optimal query complexity in all parameters, which is 
$$\Or\left(d\|H\|_{\max} + \frac{\log(1/\epsilon)}{\log\log(1/\epsilon)} \right). $$
Here $d$ is the sparsity of the Hamiltonian, $\|H\|_{\max}$ denotes the largest matrix element of $H$ in absolute value. 
Notice that $\|H_1\|_{\max} = \Or(n^2)$ since every non-zero entry of $H_1$ is either $n^2$ or $(-2n^2)$, and $\|H_2\|_{\max} = \Or(1)$, we have $\|H\|_{\max} = \Or(n^2)$. 
Therefore the query complexity becomes 
$$\Or\left(n^2 + \frac{\log(n^2/\epsilon)}{\log\log(n^2/\epsilon)} \right) = {\Or}\left(n^2 + \log(n^2/\epsilon)\right) = \Or\left(\frac{1}{\epsilon} + \log(1/\epsilon)\right) = \Or\left(\frac{1}{\epsilon}\right).$$

\paragraph{Interaction picture} \cite[Theorem 7]{LowWiebe2019} shows that by applying truncated Dyson series to simulate time-independent Hamiltonian $H_1+H_2$ in the interaction picture rather than the Schr\"odinger picture, it requires
$$\Or\left(\|H_2\| \frac{\log(\|H_2\|/\epsilon)}{\log\log(\|H_2\|/\epsilon)}\right)$$
queries to $H_2$ and 
$$\Or\left(\|H_2\| \frac{\log(\|H_2\|/\epsilon)}{\log\log(\|H_2\|/\epsilon)} \log\left(\frac{\|H_1\| +\|H_2\|}{\epsilon}\right)\right)$$
queries to the unitary time evolution $e^{-\I s H_1 }$. Therefore the query complexity is logarithmic in $n$ and thus the scaling in terms of $\epsilon$ is still poly-logarithmic. Note that the number of time steps is included in the oracle HAM-T and scales as $\Or(\norm{H_1})$ \cite[Lemma 6]{LowWiebe2019}.

\subsection{Time-dependent schemes}

\paragraph{Time-dependent second order Trotter formulae} \cite[Eq. (A12-A14)]{HuyghebaertDeRaedt1990} show that for generalized second-order Trotter formula applied to the model \cref{eqn:schrodinger_tdmass} with time-independent mass and time-dependent frequency (in particular, $f_2(t)H_2$ and $f_2(s)H_2$ commute for any $t$ and $s$), the one-step local Trotter error scales as 
$$\Or\left(h^3 \left(\|[H_1,H_2]\|+\|[H_1,[H_1,H_2]]\|+\|[H_2,[H_2,H_1]]\|\right)\right),$$
thus the global error scales 
$$\Or\left(h^2\left(\|[H_1,H_2]\|+\|[H_1,[H_1,H_2]]\|+\|[H_2,[H_2,H_1]]\|\right)\right) = \Or\left(\frac{n^2}{L^2}\right).$$
To bound this by $\epsilon$, it suffices to choose 
$$L = \Or\left(\frac{n}{\epsilon^{1/2}}\right) = \Or\left(\frac{1}{\epsilon}\right). $$

The second order complexity estimate from~\cite{WiebeBerryHoyerEtAl2010} is a special case of their general high order result. 
We will show the general case later. 

\cite[Appendix A]{WeckerHastingsWiebeEtAl2015} proves an improved operator norm error bound for the second order standard Trotter formula. 
The one-step local Trotter error is bounded by 
$$\left(\frac{1}{24}\sup\|H''(s)\| + \frac{1}{12} \sup\|[H'(s),H(s)]\| + \sup \|[H_1(s),[H_1(s),H_2(s)]\| + \|[H_2(s),[H_2(s),H_1(s)]\|\right)h^3,$$
thus the global error scales 
$$\Or\left(h^2\left(\|H_1\|+\|[H_1,H_2]\|+\|[H_1,[H_1,H_2]]\|+\|[H_2,[H_2,H_1]]\|\right)\right) = \Or\left(\frac{n^2}{L^2}\right).$$
To bound this by $\epsilon$, it suffices to choose 
$$L = \Or\left(\frac{n}{\epsilon^{1/2}}\right) = \Or\left(\frac{1}{\epsilon}\right).$$

\paragraph{Time-dependent high order Trotter formula} \cite[Theorem 1]{WiebeBerryHoyerEtAl2010} proves that, to simulate a system with Hamiltonian $H(t) = \sum_{j=1}^m H_j(t)$ within operator spectral norm error $\epsilon$ using a $2k$-th order standard Trotter formula, the total number of exponentials is 
$$2m 5^{k-1}\left\lceil 5k \Lambda T \left(\frac{5}{3}\right)^k\left(\frac{\Lambda T}{\epsilon}\right)^{1/(2k)}\right\rceil$$
where 
$$\Lambda = \sup_{p = 0,1,\cdots,2k} \left(\sup_t \left(\sum_{j=1}^m \|\partial_t^p H_j(t)\|\right)^{1/(p+1)}\right).$$
We first notice that the total number of exponentials only differ from the total number of Trotter steps by a factor of $2m5^{k-1}$. After absorbing all the terms independent of $n$ and $\epsilon$ into the big-$\Or$ notation, in the case of the Schr\"odinger equation with a time-dependent effective mass, the total number of Trotter steps becomes 
$$\Or\left(\Lambda \left(\frac{\Lambda}{\epsilon}\right)^{1/2k}\right).$$
It remains to estimate the scaling of $\Lambda$. 
By noticing $\partial_t^p H_j(t) = f_j^{(p)}(t) H_j$, we obtain that  $\left(\sum_{j=1}^m \|\partial_t^p H_j(t)\|\right)$ is dominated by $H_1 = \Or(n^2)$, and 
$$\Lambda = \Or\left(\sup_{p=0,1,\cdots,2k} (n^2)^{1/(p+1)}\right) = \Or\left(n^2\right).$$
Therefore the total number of Trotter steps becomes 
$$\Or\left(\frac{n^{2+1/k}}{\epsilon^{1/(2k)}}\right) =\Or\left(\frac{1}{\epsilon^{1+1/k}}\right).$$

\paragraph{Truncated Dyson series} \cite[Theorem 9]{LowWiebe2019} shows that to obtain an $\epsilon$ approximation of the exact evolution operator with success probability at least $(1-\epsilon)$ using truncated Dyson series method, the query complexity is 
\REV{
$$\Or\left(d \|H\|_{\max,\infty} T \frac{\log(d \|H\|_{\max,\infty} T/\epsilon)}{\log\log(d \|H\|_{\max,\infty} T/\epsilon)}\right). $$
Here $d$ is the sparsity of the Hamiltonian, and $\|H\|_{\max,\infty} = \sup_{t\in[0,T]} \|H(t)\|_{\max}$, where $\|A\|_{\max}$ denotes the largest matrix element of $A$ in absolute value. 
In the case of the model \cref{eqn:schrodinger_tdmass}, noticing that $\|H_1\|_{\max} = \Or(n^2)$ because every non-zero entry of $H_1$ is either $n^2$ or $(-2n^2)$, we have $\|H(t)\|_{\max,\infty} = \Or(\|H_1\|_{\max}) = \Or(n^2)$, then the query complexity becomes 
$$\Or\left(n^2\frac{\log(n^2/\epsilon)}{\log\log(n^2/\epsilon)} \right) = \widetilde{\Or}\left(n^2\right) = \widetilde{\Or}\left(\frac{1}{\epsilon}\right).$$}

\paragraph{Rescaled Dyson series} \cite[Theorem 10]{BerryChildsSuEtAl2020} shows that to obtain an $\epsilon$ approximation of the exact evolution operator using rescaled Dyson series method, the query complexity is 
$$\Or\left(d \|H\|_{\max,1} \frac{\log(d \|H\|_{\max,1}/\epsilon)}{\log\log(d \|H\|_{\max,1}/\epsilon)}\right).$$
Here $d$ is the sparsity of the Hamiltonian, $\|H\|_{\max,1} = \int_0^T \|H(t)\|_{\max} \, dt$ where $\|A\|_{\max} $ denotes the largest matrix element of $A$ in absolute value. 
In the case of the model \cref{eqn:schrodinger_tdmass}, noticing that $\|H_1\|_{\max} = \Or(n^2)$ because every non-zero entry of $H_1$ is either $n^2$ or $(-2n^2)$, we have $\|H\|_{\max,1} = \Or(n^2)$. 
Therefore the query complexity becomes 
$$\Or\left(n^2\frac{\log(n^2/\epsilon)}{\log\log(n^2/\epsilon)} \right) = \widetilde{\Or}\left(n^2\right) = \widetilde{\Or}\left(\frac{1}{\epsilon}\right).$$

We mention that in~\cite{BerryChildsSuEtAl2020} another method called continuous qDRIFT is also proposed to successfully achieve $L^1$ scaling of the Hamiltonian.
However, continuous qDRIFT is a first order method, and its complexity dependence on $\|H\|_{\max,1}$ is quadratic, which is worse than that of rescaled Dyson series. 
Hence we only include the rescaled Dyson series method in our table for comparison.

\section{Proof of error representations}\label{append:proof_error_represent}

In this part, we derive the error representations of the first-order and second-order Trotter formulae, as presented in \cref{lem:error_rep_s_1st} - \cref{lem:error_rep_g_2nd}. All of the proofs consisting of the following two steps: One first compares the derivatives
\begin{equation} \label{eqn:u_h}
        \partial_h U(h,0) = (-\I f_1(h)H_1-\I f_2(h)H_2) U(h,0),
\end{equation} 
and its numerical analogs $\partial_h U_{m,p}(h,0)$ ($m = g, s$ and $p = 1, 2$), and apply the variation of parameter formula (\cref{lem:VoP}); Then the Taylor theorem (\cref{lem:Taylor}) is applied to further simplify the terms.

We first present the proof of \cref{lem:error_rep_g_1st}, since its error representation contains fewest terms. The rest of the error representations, \cref{lem:error_rep_s_1st}, \cref{lem:error_rep_s_2nd} and \cref{lem:error_rep_g_2nd}, follow the exact same idea of proof, just involving more calculations.

\begin{proof}[Proof of \cref{lem:error_rep_g_1st}] By taking derivative of $U_{g,1}(h,0)$ with respect to $h$, one has
\begin{equation} \label{eqn:u_h_g1}
        \begin{split}
            \partial_h U_{g,1}(h,0) &= -\I f_2(h) H_2 U_{g,1}(h,0) + \exp\left(-\I \int_0^h f_2(s)ds H_2\right)\left(-\I f_1(h) H_1\right)\exp\left(-\I \int_0^h f_1(s)ds H_1\right)\\
            & = \left(-\I f_1(h)H_1 - \I f_2(h)H_2\right) U_{g,1}(h,0) \\
            & \quad +  \exp\left(-\I \int_0^h f_2(s)ds H_2\right) E_{g,1}(h)\exp\left(-\I \int_0^h f_1(s)ds H_1\right)
        \end{split}
    \end{equation}
    where $E_{g,1}(h)$ is defined as
    \begin{equation}
        E_{g,1}(h) = \I f_1(h) \left[ \exp\left(\ad_{\I \int_0^h f_2(s')ds' H_2}\right)H_1 - H_1\right]. 
    \end{equation}
    By applying \cref{lem:VoP} to \cref{eqn:u_h} and \cref{eqn:u_h_g1}, one obtains 
    \begin{equation}
        U_{g,1}(h,0) = U(h,0) + \int_0^h U(h,s) \exp\left(-\I \int_0^s f_2(s')ds' H_2\right) E_{g,1}(s)\exp\left(-\I \int_0^s f_1(s')ds' H_1\right) ds. 
    \end{equation}
    The representation of $E_{g,1}$, by Taylor's theorem (\cref{lem:Taylor}), reads
    \begin{equation}
        \begin{split}
             E_{g,1}(h) &= \int_0^h f_1(h) f_2(s) \left(\exp\left(\ad_{\I \int_0^h f_2(s')ds' H_2}\right)([H_1,H_2])\right) ds. 
        \end{split}
    \end{equation}
    
\end{proof}

\begin{proof}[Proof of \cref{lem:error_rep_s_1st}]
One starts by taking derivative of $U_{s,1}(h,0)$ with respect to $h$, which reads
    \begin{equation}\label{eqn:u_h_s1}
        \begin{split}
            \partial_h U_{s,1}(h,0) &= \left(-\I f_2(h)H_2 - \I h f_2'(h)H_2\right)\exp\left(-\I h f_2(h)H_2\right)\exp\left(-\I h f_1(h)H_1\right) \\
            & \quad + \exp\left(-\I h f_2(h)H_2\right)\left(-\I f_1(h)H_1 - \I h f_1'(h)H_1\right)\exp\left(-\I h f_1(h)H_1\right)\\
            & = \left(-\I f_2(h)H_2 - \I f_1(h)H_1\right) U_{s,1}(h,0) \\
            & \quad + \exp\left(-\I h f_2(h)H_2\right) E_{s,1}(h) \exp\left(-\I h f_1(h)H_1\right),
        \end{split}
    \end{equation}
where $E_{s,1}(h)$ is defined as
    \begin{equation}
    \begin{split}
        E_{s,1}(h) =  \I f_1(h) \left[\exp\left(\ad_{\I h f_2(h) H_2}\right)H_1 - H_1\right] 
        - \I h f_1'(h)H_1 - \I h f_2'(h)H_2. 
    \end{split}
    \end{equation}
By applying \cref{lem:VoP} to \cref{eqn:u_h} and \cref{eqn:u_h_s1}, one has
    \begin{equation}
        U_{s,1}(h,0) = U(h,0) + \int_{0}^h U(h,s) \exp\left(-\I s f_2(s)H_2\right) E_{s,1}(s) \exp\left(-\I s f_1(s)H_1\right) ds.
    \end{equation}
It remains to derive the representation of $E_{s,1}$. The representation of $E_{s,1}$ can be derived from Taylor's theorem up to first-order. 
    By \cref{lem:Taylor}
    \begin{equation}
        \begin{split}
            & \quad \exp\left(\ad_{\I h f_2(h) H_2}\right)H_1 - H_1 \\
            & = \int_{0}^h \I f_2(s) \left(\exp\left(\ad_{\I s f_2(s) H_2}\right)([H_2,H_1])\right)ds + \int_{0}^h \I s f_2'(s) \left(\exp\left(\ad_{\I s f_2(s) H_2}\right)([H_2,H_1])\right)ds.
        \end{split}
    \end{equation}
    Therefore, one has
    \begin{equation}
        \begin{split}
            E_{s,1}(h) & =  \int_{0}^h f_1(h)f_2(s) \left(\exp\left(\ad_{\I s f_2(s) H_2}\right)([H_1,H_2])\right)ds - \I h f_1'(h)H_1 - \I h f_2'(h)H_2 \\
            & \quad + \int_{0}^h  s f_1(h) f_2'(s) \left(\exp\left(\ad_{\I s f_2(s) H_2}\right)([H_1,H_2])\right)ds.
        \end{split}
    \end{equation}
\end{proof}

Before proceeding, we first define the following quantities needed in the error representations of the second order standard and generalized Trotter formulae
     \begin{equation}\label{eqn:Es2}
        \begin{split}
            E_{s,2}(h) &= \I \int_0^h f_1''(s) (h-s) H_1 ds -  \frac{\I}{8}\int_0^h f_1''(s/2) (2h-s) H_1 ds - \frac{\I}{4} \int_0^h f_2''(s/2) (2h-s) H_2 ds \\
            & \quad - \frac{\I}{8} \int_0^h  \left[f_1''(s/2)\exp\left(\ad_{-\I s f_2(s/2)H_2}\right)H_1\right](2h-s)ds \\
            & \quad + \frac{1}{4} \int_0^h \left[ f_1'(s/2)f_2(s/2) \left(\exp\left(\ad_{-\I s f_2(s/2)H_2 }\right)H_1\right)\right] h ds \\
            & \quad + \I \int_0^h \left[ f_2''(s)\exp\left(\ad_{i\frac{s}{2}f_1(s/2)H_1}\right) H_2\right](h-s) ds \\
            & \quad +\frac{1}{2}\int_0^h  (f_1'(s/2)f_2(s/2)+f_1(s/2)f_2'(s/2)) \left(\exp\left(\ad_{-\I s f_2(s/2)H_2}\right)[H_1,H_2]\right) (h-s)ds \\
            & \quad - \frac{1}{2} \int_0^h \left(f_2'(s)f_1(s/2)+f_2(s)f_1'(s/2) \right)\left(\exp\left(\ad_{i\frac{s}{2}f_1(s/2)H_1}\right) [H_1,H_2]\right)(h-s)ds  \\
            & \quad +\frac{\I}{2} \int_0^h \left[f_1(s/2)f_2^2(s/2) \left(\exp\left(\ad_{-\I s f_2(s/2)H_2}\right) [H_2,[H_1,H_2]]\right)\right](h-s)ds \\
            & \quad - \frac{\I}{4} \int_0^h \left[f_2(s)f_1^2(s/2) \left(\exp\left(\ad_{i\frac{s}{2}f_1(s/2)H_1}\right) [H_1,[H_1,H_2]]\right) \right] (h-s)ds \\
            & \quad + \frac{1}{8} \int_0^h \left[ f_1'(s/2)f_2'(s/2) \left(\exp\left(\ad_{-\I s f_2(s/2)H_2}\right)H_1\right) \right]sh ds  \\
            & \quad +\frac{1}{4}\int_0^h  \left( f_1'(s/2)f_2'(s/2) + \frac{1}{2}f_1(s/2)f_2''(s/2)\right) \left(\exp\left(\ad_{-\I s f_2(s/2)H_2}\right) [H_1,H_2]\right) s(h-s)ds  \\
            & \quad - \frac{1}{4}\int_0^h \left( f_2'(s)f_1'(s/2) + \frac{ 1}{2}f_2(s)f_1''(s/2)\right)\exp\left(\ad_{i\frac{s}{2}f_1(s/2)H_1}\right) [H_1,H_2]s(h-s)ds  \\
            & \quad +\frac{\I}{2} \int_0^h \left[f_1(s/2) f_2(s/2)f_2'(s/2) \left(\exp\left(\ad_{-\I s f_2(s/2)H_2}\right) [H_2,[H_1,H_2]]\right)\right]s(h-s)ds \\
            & \quad - \frac{\I}{4} \int_0^h \left[f_2(s)f_1(s/2)f_1'(s/2) \left(\exp\left(\ad_{i\frac{s}{2}f_1(s/2)H_1}\right) [H_1,[H_1,H_2]]\right) \right] s(h-s)ds \\
            & \quad +\frac{\I}{8} \int_0^h \left[f_1(s/2)f_2'^2(s/2) \left(\exp\left(\ad_{-\I s f_2(s/2)H_2}\right) [H_2,[H_1,H_2]]\right)\right]s^2(h-s)ds \\
            & \quad - \frac{\I}{16} \int_0^h \left[f_2(s)f_1'^2(s/2) \left(\exp\left(\ad_{i\frac{s}{2}f_1(s/2)H_1}\right) [H_1,[H_1,H_2]]\right) \right] s^2(h-s)ds,  
        \end{split}
    \end{equation}
and
    \begin{equation}\label{eqn:Eg2}
        \begin{split}
            E_{g,2}(h) &= -\frac{h}{2}f_1(0)\int_0^h f_2'(s)ds [H_1,H_2] + 
        \frac{h}{4}f_2(0)\int_0^h f_1'(s/2)ds[H_1,H_2] \\
        & \quad - f_2(h)\int_0^h( f_1'(s)-\frac{1}{4} f_1'(s/2))
        \left(\exp\left(\ad_{\I\int_{s/2}^{s}f_1(s')ds' H_1}\right)[H_1,H_2]\right)(h-s)ds \\
         & \quad + \frac{1}{2}f_1(h/2) \int_0^h  f_2'(s)\left(\exp\left(\ad_{-\I\int_{0}^{s}f_2(s')ds' H_2}\right)[H_1,H_2]\right)(h-s)ds \\
        & \quad - \I f_2(h)\int_0^h\left( f_1(s)-\frac{1}{2}f_1(s/2)\right)^2
        \left(\exp\left(\ad_{\I\int_{s/2}^{s}f_1(s')ds' H_1}\right)[H_1,[H_1,H_2]]\right)(h-s)ds\\
        &\quad + \frac{\I}{2}f_1(h/2)\int_0^h  f_2^2(s)\left(\exp\left(\ad_{-\I\int_{0}^{s}f_2(s')ds' H_2}\right)[H_2,[H_2,H_1]]\right)(h-s)ds.
        \end{split}
    \end{equation}

\begin{proof}[Proof of \cref{lem:error_rep_s_2nd}]
One first compute the derivative with respect to $h$ of $U_{s,2}$
    \begin{equation}
        \begin{split}
            \partial_h U_{s,2} &= \left(-\I \frac{1}{2}f_1(h/2)- \I\frac{h}{4}f_1'(h/2)\right)H_1\exp\left(-\I\frac{h}{2}f_1(h/2)H_1\right)\exp\left(-\I h f_2(h/2)H_2\right)\exp\left(-\I\frac{h}{2}f_1(h/2)H_1\right) \\
            & \quad + \exp\left(-\I\frac{h}{2}f_1(h/2)H_1\right)\left(-\I f_2(h/2)-\I \frac{h}{2} f_2'(h/2)\right)H_2\exp\left(-\I h f_2(h/2)H_2\right)\exp\left(-\I\frac{h}{2}f_1(h/2)H_1\right) \\
            & \quad + \exp\left(-\I\frac{h}{2}f_1(h/2)H_1\right)\exp\left(-\I h f_2(h/2)H_2\right)\left(-\I \frac{1}{2}f_1(h/2)- \I\frac{h}{4}f_1'(h/2)\right)H_1\exp\left(-\I\frac{h}{2}f_1(h/2)H_1\right) \\
            & = \left(-\I f_1(h)H_1 - \I f_2(h)H_2\right)U_{s,2} \\
            & \quad + \exp\left(-\I\frac{h}{2}f_1(h/2)H_1\right)E_{s,2}(h)\exp\left(-\I h f_2(h/2)H_2\right)\exp\left(-\I\frac{h}{2}f_1(h/2)H_1\right), 
        \end{split}
    \end{equation}
where $E_{s,2}(h)$ is defined as
    \begin{equation}
        \begin{split}
            E_{s,2}(h) &= \I f_1(h)H_1 - \I\frac{1}{2}f_1(h/2)H_1 -  \I \frac{h}{4}f_1'(h/2) H_1 -\left(\I\frac{1}{2}f_1(h/2)+\I\frac{h}{4}f_1'(h/2)\right)\exp\left(\ad_{-\I h f_2(h/2)H_2 }\right)H_1 \\
            & \quad + \I f_2(h) \exp\left(\ad_{i\frac{h}{2}f_1(h/2)H_1}\right) H_2 - \left(\I f_2(h/2) + \frac{\I h}{2}f_2'(h/2)\right)H_2. 
        \end{split}
    \end{equation}
    Similar as the proofs for first-order formulae, applying \cref{lem:VoP} gives
    \begin{equation}
        U_{s,2}(h,0) = U(h,0) + \int_0^{h} U(h,s) \exp\left(-\I\frac{s}{2}f_1(s/2)H_1\right)E_{s,2}(s)\exp\left(-\I s f_2(s/2)H_2\right)\exp\left(-\I\frac{s}{2}f_1(s/2)H_1\right) ds. 
    \end{equation}
    The rest of the proof follows straightforward calculations. To be exact, one then applies the Taylor's theorem (\cref{lem:Taylor}) to expand each term in $E_{s,2}$ to second-order in terms of $h$ with respect to 0. 
    The first three terms can be expressed as
    \begin{equation}
        \I f_1(h) H_1 = \I f_1(0) H_1 + \I h f_1'(0) H_1 + \I \int_0^h f_1''(s) (h-s) H_1 ds, 
    \end{equation}
    \begin{equation}
        -\I \frac{1}{2}f_1(h/2) H_1 = -\I \frac{1}{2} f_1(0) H_1 - \I \frac{h}{4} f_1'(0) H_1 - \I \frac{1}{8}\int_0^h f_1''(s/2) (h-s) H_1 ds, 
    \end{equation}
    \begin{equation}
        -  \I \frac{h}{4}f_1'(h/2) H_1 = -  \I \frac{h}{4}f_1'(0)H_1 -  \I \frac{h}{8} \int_0^h f_1''(s/2) H_1 ds, 
    \end{equation}
    Similarly, let us apply the Taylor theorem to the fourth term in $E_{s,2}$, which is the sum of
    \begin{equation}
        \begin{split}
            & -\I\frac{1}{2}f_1(h/2)\exp\left(\ad_{-\I h f_2(h/2)H_2 }\right)H_1 
            \\
            = & -\frac{\I}{2} f_1(0)H_1 - \frac{\I h}{4}f_1'(0)H_1 + \frac{h}{2}f_1(0)f_2(0)[H_1,H_2] 
            - \frac{\I}{2} \int_0^h ds \frac{1}{4}f_1''(s/2)\exp\left(\ad_{-\I s f_2(s/2)H_2}\right)H_1 (h-s)
            \\
            &- \frac{\I}{2} \int_0^h ds f_1(s/2)\left(\I f_2(s/2)+\frac{\I s}{2}f_2'(s/2)\right)^2 \exp\left(\ad_{-\I s f_2(s/2)H_2}\right) [H_2,[H_1,H_2]] (h-s) 
            \\
            & - \frac{\I}{2} \int_0^h ds \left(\I(f_1'(s/2)f_2(s/2)+f_1(s/2)f_2'(s/2))+\frac{\I s}{2}f_1'(s/2)f_2'(s/2) + \frac{\I s}{4}f_1(s/2)f_2''(s/2)\right)  \\
            & \quad\quad\quad\quad\quad\quad\quad\quad \times \exp\left(\ad_{-\I s f_2(s/2)H_2}\right) [H_1,H_2] (h-s),  
        \end{split}
    \end{equation}
    and
    \begin{equation}
        \begin{split}
            &-\I\frac{h}{4}f_1'(h/2)\exp\left(\ad_{-\I h f_2(h/2)H_2 }\right)H_1 \\
            = & -\frac{\I h}{4} f_1'(0) H_1 
             - \frac{\I h}{4} \int_0^h ds  \frac{1}{2}f_1''(s/2)\exp\left(\ad_{-\I s f_2(s/2)H_2 }\right)H_1 \\
            &  - \frac{\I h}{4} \int_0^h ds \I f_1'(s/2)\left(f_2(s/2) + \frac{s}{2}f_2'(s/2)\right) \exp\left(\ad_{-\I s f_2(s/2)H_2 }\right)H_1  , 
        \end{split}
    \end{equation}
  The fifth term in $E_{s,2}$ reads
    \begin{equation}
        \begin{split}
            &  \I f_2(h) \exp\left(\ad_{i\frac{h}{2}f_1(h/2)H_1}\right) H_2 \\
            & = \I f_2(0) H_2 + \I h f_2'(0) H_2 - \frac{h}{2} f_2(0)f_1(0) [H_1,H_2] \\
            &   + \I \int_0^h ds (h-s)
            f_2''(s)\exp\left(\ad_{i\frac{s}{2}f_1(s/2)H_1}\right) H_2 (h-s) \\
            & + \I \int_0^h ds (h-s) f_2(s)\left(\frac{\I}{2}f_1(s/2)+\frac{\I s}{4}f_1'(s/2)\right)^2 \exp\left(\ad_{i\frac{s}{2}f_1(s/2)H_1}\right) [H_1,[H_1,H_2]] \Big] (h-s)\\
            & + \I \int_0^h ds (h-s) \left(\frac{\I}{2}f_2'(s)f_1(s/2)+\frac{\I}{2}f_2(s)f_1'(s/2) + \frac{\I s}{4} f_2'(s)f_1'(s/2) + \frac{\I s}{8}f_2(s)f_1''(s/2)\right)   \\
            & \quad\quad\quad\quad\quad\quad\quad\quad\quad\quad\times \exp\left(\ad_{i\frac{s}{2}f_1(s/2)H_1}\right) [H_1,H_2] (h-s), 
        \end{split}
    \end{equation}
    The last term in $E_{s,2}$ is the sum of
    \begin{equation}
        -\I f_2(h/2)H_2 = -\I f_2(0)H_2 - \I \frac{h}{2} f_2'(0)H_2 - \I \int_0^h \frac{1}{4}f_2''(s/2) (h-s)ds, 
    \end{equation}
    and 
    \begin{equation}
        -\frac{\I h}{2}f_2'(h/2)H_2 = -\frac{\I h}{2} f_2'(0)H_2 - \frac{\I h}{4}\int_0^h f_2''(s/2)H_2ds. 
    \end{equation}
    Notice that, if we add the above eight equations together, all the zeroth-order and first-order terms of $h$ cancel, then the desired expression of $ E_{s,2}(h)$ is achieved.

\end{proof}

\begin{proof}[Proof of \cref{lem:error_rep_g_2nd}]
    The strategy for proving \cref{lem:error_rep_g_2nd} is the same as that for \cref{lem:error_rep_s_2nd}. 
    By taking derivatives with respect to $h$ in both $U(h,0)$ and $U_{g,2}(h,0)$, we have \cref{eqn:u_h} and 
    \begin{equation}
        \begin{split}
            \partial_h U_{g,2}(h,0) &= \left(-\I f_1(h)+ \frac{\I}{2} f_1(h/2)\right)H_1U_{g,2}(h,0)\\
        & \quad + \exp\left(-\I\int_{h/2}^{h}f_1(s)dsH_1\right)(-\I f_2(h)H_2) \exp\left(-\I\int_{0}^{h}f_2(s)dsH_2\right)\exp\left(-\I\int_{0}^{h/2}f_1(s)dsH_1\right) \\
        & \quad + \exp\left(-\I\int_{h/2}^{h}f_1(s)dsH_1\right)\exp\left(-\I\int_{0}^{h}f_2(s)dsH_2\right)\\
        & \quad\quad \times \left(-\frac{\I}{2} f_1(h/2)/H_1\right) \exp\left(-\I\int_{0}^{h/2}f_1(s)dsH_1\right) \\
        & = -(\I f_1(h)H_1 + \I f_2(h)H_2)U_{g,2}(h,0) \\
        & \quad\quad + \exp\left(-\I\int_{h/2}^h f_1(s)ds H_1\right)E_{g,2}(h)\exp\left(-\I \int_0^h f_2(s)ds H_2\right)\exp\left(-\I\int_0^{h/2}f_1(s)ds H_1\right), 
        \end{split}
    \end{equation}
    where $E_{g,2}(h)$ denotes
    \begin{equation}
        \begin{split}
            E_{g,2}(h) &= if_2(h)\exp(\ad_{\I\int_{h/2}^{h}f_1(s)dsH_1})H_2 
        + \frac{\I}{2} f_1(h/2)H_1  \\
        & \quad - \I f_2(h)H_2 - \frac{\I}{2} f_1(h/2)\exp\left(\ad_{-\I\int_{0}^{h}f_2(s)dsH_2}\right)H_1 \\
        &= \I f_2(h)\left[\exp\left(\ad_{\I\int_{h/2}^{h}f_1(s)dsH_1}\right)H_2 - H_2\right] \\
        & \quad - \frac{\I}{2}f_1(h/2)\left[\exp\left(\ad_{-\I\int_{0}^{h}f_2(s)dsH_2}\right)H_1-H_1\right]. 
        \end{split}
    \end{equation}
     By applying \cref{lem:VoP}, we have 
    \begin{equation}
    \begin{split}
        U_{g,2}(h,0) = U(h,0) + & \int_0^{h} U(h,s) \exp\left(-\I\int_{s/2}^s f_1(s')ds' H_1\right)E_{g,2}(s) \\
        & \quad \times \exp\left(-\I \int_0^s f_2(s')ds' H_2\right)\exp\left(-\I\int_0^{s/2}f_1(s')ds' H_1\right) ds. 
    \end{split}
    \end{equation}
    It remains to derive the representation of $E_{g,2}$. It follows from the Taylor's theorem (\cref{lem:Taylor}) that
    \begin{align*}
        &\quad \exp\left(\ad_{\I \int_{h/2}^{h}f_1(s)dsH_1}\right)H_2 - H_2 \\
        &= \frac{\I h}{2}f_1(0)[H_1,H_2] + \int_0^h\left(\I f_1'(s)-\frac{1}{4}\I f_1'(s/2)\right)
        \left(\exp\left(\ad_{\I\int_{s/2}^{s}f_1(s')ds' H_1}\right)[H_1,H_2]\right)(h-s)ds \\
        & \quad\quad  + \int_0^h\left(\I f_1(s)-\frac{1}{2}\I f_1(s/2)\right)^2
        \left(\exp\left(\ad_{\I\int_{s/2}^{s}f_1(s')ds'H_1}\right)[H_1,[H_1,H_2]]\right)(h-s)ds,
    \end{align*}
    and
    \begin{align*}
        &\quad \exp\left(\ad_{-\I\int_{0}^{h}f_2(s)ds H_2}\right)H_1 - H_1 \\
        &= -\I h f_2(0)[H_2,H_1] - \int_0^h \I f_2'(s)\left(\exp\left(\ad_{-\I\int_{0}^{s}f_2(s')ds' H_2}\right)[H_2,H_1]\right)(h-s)ds \\
        &\quad\quad + \int_0^h (\I f_2(s))^2\left(\exp\left(\ad_{-\I\int_{0}^{s}f_2(s')ds' H_2}\right)[H_2,[H_2,H_1]]\right)(h-s)ds. 
    \end{align*}
    Thus we have
    \begin{align*}
        E_{g,2}(h) &= \I f_2(h)\frac{\I h}{2}f_1(0)[H_1,H_2] 
        + \frac{\I}{2}f_1(h/2)\I h f_2(0)[H_2,H_1] \\
        & \quad + \I f_2(h)\int_0^h\left(\I f_1'(s)-\frac{\I}{4}f_1'(s/2)\right)
        \left(\exp\left(\ad_{\I\int_{s/2}^{s}f_1(s')ds'H_1}\right)[H_1,H_2]\right)(h-s)ds \\
        & \quad + \I f_2(h)\int_0^h\left(\I f_1(s)-\frac{\I}{2}f_1(s/2)\right)^2
        \left(\exp\left(\ad_{\I\int_{s/2}^{s}f_1(s')ds'H_1}\right)[H_1,[H_1,H_2]]\right)(h-s)ds\\
        & \quad +\frac{\I}{2}f_1(h/2) \int_0^h \I f_2'(s)\left(\exp\left(\ad_{-\I\int_{0}^{s}f_2(s')ds' H_2}\right)[H_2,H_1]\right)(h-s)ds \\
        &\quad - \frac{\I}{2}f_1(h/2)\int_0^h (\I f_2(s))^2\left(\exp\left(\ad_{-\I\int_{0}^{s}f_2(s')ds' H_2}\right)[H_2,[H_2,H_1]]\right)(h-s)ds \\
        &=  -\frac{h}{2}f_1(0)\int_0^h f_2'(s)ds [H_1,H_2] + 
        \frac{h}{4}f_2(0)\int_0^h f_1'(s/2)ds[H_1,H_2] \\
        & \quad - f_2(h)\int_0^h( f_1'(s)-\frac{1}{4} f_1'(s/2))
        \left(\exp\left(\ad_{\I\int_{s/2}^{s}f_1(s')ds' H_1}\right)[H_1,H_2]\right)(h-s)ds \\
         & \quad + \frac{1}{2}f_1(h/2) \int_0^h  f_2'(s)\left(\exp\left(\ad_{-\I\int_{0}^{s}f_2(s')ds' H_2}\right)[H_1,H_2]\right)(h-s)ds \\
        & \quad - \I f_2(h)\int_0^h\left( f_1(s)-\frac{1}{2}f_1(s/2)\right)^2
        \left(\exp\left(\ad_{\I\int_{s/2}^{s}f_1(s')ds' H_1}\right)[H_1,[H_1,H_2]]\right)(h-s)ds\\
        &\quad + \frac{\I}{2}f_1(h/2)\int_0^h  f_2^2(s)\left(\exp\left(\ad_{-\I\int_{0}^{s}f_2(s')ds' H_2}\right)[H_2,[H_2,H_1]]\right)(h-s)ds. 
    \end{align*}

\end{proof}

\section{Proof of \cref{lem:harmonic_verification_assumption}}\label{append:proof_assumption}

\begin{proof}
    Let $V_k^{(0)} = V_k = V(x_k)$ for $0 \leq k \leq n-1$ and $V_{k+n} = V_{k} = V_{k-n}$ defined in a cyclic manner. Recursively we define $V_{k}^{(j+1)} = n(V_{k+1}^{(j)} - V_{k}^{(j)})$. 
    Notice that $V_{k}^{(j)}$ is an approximation of the $j$-th order derivative of $V(x)$ evaluated at $x = x_k$. 
    By Taylor's theorem and the assumption that $V(x)$ has bounded derivatives up to fourth order, we obtain that $V_{k}^{(j)}$ is bounded for any $k$ and $0 \leq j \leq 4$. 
    The equality $H_1 = D_1^{\dagger}D_1$ directly follows from the definition. 
    We focus on the proof of the commutator bounds. 
    
    We start with the calculation of an explicit expression of $[H_1,H_2]$, 
    \begin{equation}
        H_1H_2 = n^2 \left(\begin{array}{ccccc}
        2V_0 & -V_1 & & & -V_{n-1} \\
         -V_0& 2V_1 & -V_2 & & \\
          & \ddots& \ddots& \ddots& \\
           & & -V_{n-3}& 2V_{n-2} & -V_{n-1}\\
        -V_0& & & -V_{n-2} & 2V_{n-1} \\
    \end{array}\right), 
    \end{equation}
    \begin{equation}
        H_2H_1 = n^2 \left(\begin{array}{ccccc}
        2V_0 & -V_0 & & & -V_{0} \\
         -V_1& 2V_1 & -V_1 & & \\
          & \ddots& \ddots& \ddots& \\
           & & -V_{n-2}& 2V_{n-2} & -V_{n-2}\\
        -V_{n-1}& & & -V_{n-1} & 2V_{n-1} \\
    \end{array}\right). 
    \end{equation}
    Then 
    \begin{equation}
    \begin{split}
        [H_1,H_2] &= n^2 \left(\begin{array}{ccccc}
        0 & V_0-V_1 & & & V_0-V_{n-1} \\
         V_1-V_0& 0 & V_1-V_2 & & \\
          & \ddots& \ddots& \ddots& \\
           & & V_{n-2}-V_{n-3}& 0 & V_{n-2}-V_{n-1}\\
        V_{n-1}-V_0& & & V_{n-1}-V_{n-2} & 0 \\
    \end{array}\right) \\
    & = n \left(\begin{array}{ccccc}
        0 & -V_0^{(1)} & & & V_{n-1}^{(1)} \\
         V_0^{(1)}& 0 & -V_1^{(1)} & & \\
          & \ddots& \ddots& \ddots& \\
           & & V_{n-3}^{(1)}& 0 & -V_{n-2}^{(1)}\\
        -V_{n-1}^{(1)} & & & V_{n-2}^{(1)} & 0 \\
    \end{array}\right).
    \end{split}
    \end{equation}
    We further split $[H_1,H_2] = D_L + D_R + S$ where 
    \begin{equation}
        D_L = n \left(\begin{array}{ccccc}
        -V_{n-1}^{(1)} &  & & & V_{n-1}^{(1)} \\
         V_0^{(1)}& -V_0^{(1)} &  & & \\
          & \ddots& \ddots& & \\
           & & V_{n-3}^{(1)}& -V_{n-3}^{(1)} & \\
         & & & V_{n-2}^{(1)} & -V_{n-2}^{(1)} \\
    \end{array}\right), 
    \end{equation}
    \begin{equation}
        D_R = n \left(\begin{array}{ccccc}
        V_0^{(1)} & -V_0^{(1)} & & &  \\
         & V_1^{(1)} & -V_1^{(1)} & & \\
          & & \ddots& \ddots& \\
           & & & V_{n-2}^{(1)} & -V_{n-2}^{(1)}\\
        -V_{n-1}^{(1)} & & &  & V_{n-1}^{(1)} \\
    \end{array}\right).
    \end{equation}
    and $S = -\text{diag}(V_{n-1}^{(2)}, V_{0}^{(2)}, V_{1}^{(2)}, \cdots, V_{n-2}^{(2)})$. 
    Notice that for any vector $\vec{v} = (v_k)_{k=0}^{n-1}$, 
    \begin{equation}
    \begin{split}
        \|D_L \vec{v}\|_{\star}^2 &= n\sum_{k=0}^{n-1} |V_{k-1}^{(1)}|^2|v_{k-1}-v_k|^2 \\
        & \leq n\sup|V_{k}^{(1)}|^2 \sum_{k=0}^{n-1}|v_{k-1}-v_k|^2 
         = n\sup|V_{k}^{(1)}|^2 \left(\vec{v}^\dagger \frac{H_1}{n^2}  \vec{v}\right) \\
        & = \frac{1}{n}\sup|V_{k}^{(1)}|^2 \left(\vec{v}^\dagger H_1  \vec{v}\right) 
         = \sup|V_{k}^{(1)}|^2 \|D_1\vec{v}\|_{\star}^2, 
    \end{split}
    \end{equation}
    and similarly $\|D_R \vec{v}\|_{\star}^2 \leq \sup|V_{k}^{(1)}|^2 \|D_1 \vec{v}\|_{\star}^2$. 
    Furthermore we have $\|S\vec{v}\|_{\star} \leq \sup|V_k^{(2)}| \|\vec{v}\|_{\star}$, thus there exists $\widetilde{C}$ such that 
    \begin{equation}
        \|[H_1,H_2]\vec{v}\|_{\star} \leq \|D_L\vec{v}\|_{\star} + \|D_R\vec{v}\|_{\star} + \|S\vec{v}\|_{\star} \leq \widetilde{C} (\|D_1\vec{v}\|_{\star} + \|\vec{v}\|_{\star}) . 
    \end{equation}
    
    To bound $\|[H_1,[H_1,H_2]]\vec{v}\|_{\star}$, we first compute $[H_1,[H_1,H_2]]$ and it gives 
    \begin{equation}
        [H_1,[H_1,H_2]] = n^2 \left(\begin{array}{ccccccc}
        -2V_{n-1}^{(2)} &  & V_0^{(2)} &  &. & V_{n-2}^{(2)} & \\
          & -2V_0^{(2)} & & V_1^{(2)} &  & &  V_{n-1}^{(2)} \\
         V_0^{(2)} &  & -2V_2^{(2)} &  & & & \\
           & V_1^{(2)} & & -2V_3^{(2)} &  & & \\
              & & \ddots & \ddots & \ddots & \ddots & V_{n-3}^{(2)} \\
        V_{n-2}^{(2)} & & & & & -2V_{n-3}^{(2)} & \\
         & V_{n-1}^{(2)} & & & V_{n-3}^{(2)} & & -2V_{n-2}^{(2)}
    \end{array}\right), 
    \end{equation}
    \emph{i.e. } the only non-zero entries are
    $$[H_1,[H_1,H_2]]_{k+1,k+1} = -2V_{k}^{(2)}, \quad [H_1,[H_1,H_2]]_{k,k+2} = [H_1,[H_1,H_2]]_{k+2,k} = V_{k}^{(2)}. $$
    Then we split $[H_1,[H_1,H_2]] = H_L+H_R+2H_C+2D_{DL}+2D_{DR}+W$ where 
    \begin{equation}
        H_L = n^2 \left(\begin{array}{ccccc}
        V_{n-2}^{(2)} &  &  & V_{n-2}^{(2)} & -2V_{n-2}^{(2)} \\
        -2V_{n-1}^{(2)} & V_{n-1}^{(2)} &  &  & V_{n-1}^{(2)} \\
          & & \ddots& \ddots& \\
           & & & \ddots &\\
        V_{n-3}^{(2)} &  & & -2V_{n-3}^{(2)} & V_{n-3}^{(2)} \\
    \end{array}\right), 
    \end{equation}
    \begin{equation}
        H_R = n^2 \left(\begin{array}{ccccc}
        V_0^{(2)} & -2V_0^{(2)} & V_0^{(2)} & &  \\
         & V_1^{(2)} & -2V_1^{(2)} & V_1^{(2)} & \\
          & & \ddots& \ddots& \\
           & & & \ddots &\\
        -2V_{n-1}^{(2)} & V_{n-1}^{(2)} & &  & V_{n-1}^{(2)} \\
    \end{array}\right), 
    \end{equation}
    \begin{equation}
        H_C = n^2 \left(\begin{array}{ccccc}
        -2V_{n-1}^{(2)}  & V_{n-1}^{(2)} &  & & V_{n-1}^{(2)}  \\
        V_{0}^{(2)}  & -2 V_{0}^{(2)} & V_{0}^{(2)} &  &   \\
          & & \ddots& \ddots& \\
           & & & \ddots &\\
        V_{n-2}^{(2)}  &   & & V_{n-2}^{(2)} & -2V_{n-2}^{(2)}  \\
    \end{array}\right), 
    \end{equation}
    \begin{equation}
        D_{DL }= n \left(\begin{array}{ccccc}
        V_{n-2}^{(3)} &  & & & -V_{n-2}^{(3)} \\
         -V_{n-1}^{(3)}& V_{n-1}^{(3)} &  & & \\
          & \ddots& \ddots& & \\
           & & -V_{n-4}^{(3)}& V_{n-4}^{(3)} & \\
         & & & -V_{n-3}^{(3)} & V_{n-3}^{(3)} \\
    \end{array}\right), 
    \end{equation}
     \begin{equation}
        D_{DR} = n \left(\begin{array}{ccccc}
        -V_{n-1}^{(3)} & V_{n-1}^{(3)} & & &  \\
         & -V_0^{(3)} & V_0^{(3)} & & \\
          & & \ddots& \ddots& \\
           & & & -V_{n-3}^{(3)} & V_{n-3}^{(3)}\\
        V_{n-2}^{(3)} & & &  & -V_{n-2}^{(3)} \\
    \end{array}\right),
    \end{equation}
    and $W = n^2\text{diag}(V_k^{(2)}+V_{k-2}^{(2)}-2V_{k-1}^{(2)})_{k=0}^{n-1}$. 
    Notice that $H_L,H_R,H_C$ are modifications from $H_1$, with the center of the central formula to be on the diagonal or subdiagonal and multiplying each row by different bounded parameters. 
    $D_{DL}$ and $D_{DR}$ are very similar to $D_L$ and $D_R$, with higher order potential. 
    Furthermore, we have 
    \begin{equation}
    \begin{split}
        \|H_L \vec{v}\|_{\star}^2 &= n^3 \sum_{k=0}^{n-1}|V_{k-2}^{(2)}|^2|v_k-2v_{k-1}+v_{k-2}|^2 \\
        &\leq n^3 \sup|V_{k}^{(2)}|^2 \sum_{k=0}^{n-1}|v_k-2v_{k-1}+v_{k-2}|^2 \\
        & = \sup|V_{k}^{(2)}|^2 \|H_1 \vec{v}\|_{\star}^2, 
    \end{split}
    \end{equation}
    and similarly 
    $$\|H_R \vec{v}\|_{\star} \leq \widetilde{C} \|H_1 \vec{v}\|_{\star}, \quad \|H_C \vec{v}\|_{\star} \leq \widetilde{C} \|H_1 \vec{v}\|_{\star}. $$ 
    For $D_{DL}$ and $D_{DR}$, they can be bounded by the same way we bound $D_L$ and $D_R$ before, and thus 
    $$\|D_{DL} \vec{v}\|_{\star} \leq \widetilde{C}\|D_1 \vec{v}\|_{\star} \leq \widetilde{C}(\|H_1 \vec{v}\|_{\star} + \|\vec{v}\|_{\star})$$ 
    and 
    $$\|D_{DR}\vec{v}\|_{\star}\leq \widetilde{C}(\|H_1\vec{v}\|_{\star} + \|\vec{v}\|_{\star}).$$ 
    Finally since $V$ has bounded fourth order derivative, the term $n^2(V_k^{(2)}+V_{k-2}^{(2)}-2V_{k-1}^{(2)})$ is bounded. 
    Therefore 
    $$\|W \vec{v}\|_{\star} \leq \widetilde{C} \|\vec{v}\|_{\star}.$$
    Combining all the estimates together, we have 
    \begin{equation}
    \begin{split}
        \|[H_1,[H_1,H_2]]\vec{v}\|_{\star} &\leq \|H_L\vec{v}\|_{\star} + \|H_R\vec{v}\|_{\star} + 2\|H_C\vec{v}\|_{\star} + 2\|D_{DL}\vec{v}\|_{\star} + 2\|D_{DR}\vec{v}\|_{\star} + \|W \vec{v}\|_{\star} \\
        & \leq \widetilde{C}(\|H_1 \vec{v}\|_{\star}+\|\vec{v}\|_{\star}). 
    \end{split}
    \end{equation}
\end{proof}

\section{Proof of \cref{lem:harmonic_space_error}}\label{append:proof_harmonic_space_error}

\begin{proof}
    1. From Taylor's theorem, 
    \begin{equation}
        \begin{split}
            & \quad n^2(\phi(t,x+1/n)-2\phi(t,x)+\phi(t,x-1/n)) - \Delta \phi(t,x) \\
            & = \frac{n^2}{6} \left[\int_{x-1/n}^x (y-(x-1/n))^3\frac{\partial^4}{\partial x^4}\phi(t,y) dy + \int_{x}^{x+1/n} (x+1/n-y)^3\frac{\partial^4}{\partial x^4}\phi(t,y) dy\right]. 
        \end{split}
    \end{equation}
    Therefore
    \begin{equation}
        \begin{split}
            & \quad |n^2(\phi(t,x+1/n)-2\phi(t,x)+\phi(t,x-1/n)) - \Delta \phi(t,x)| \\
            & \leq \frac{1}{6n} \left[\int_{x-1/n}^x \left|\frac{\partial^4}{\partial x^4}\phi(t,y)\right| dy + \int_{x}^{x+1/n} \left|\frac{\partial^4}{\partial x^4}\phi(t,y)\right| dy\right]\\ 
            & \leq \frac{1}{3n^2} \sup_{y\in[0,1]} \left|\frac{\partial^4}{\partial x^4}\phi(t,y)\right|. 
        \end{split}
    \end{equation}
    
    2. Since $\phi(t,x)$ satisfies the equation
    \begin{equation}
    \begin{split}
        \I \partial_t \phi(t,x) &= H(t) \phi(t,x) \\
        &= -f_1(t) n^2(\phi(t,x+1/n)-2\phi(t,x)+\phi(t,x-1/n)) + f_2(t)V(x)\phi(t,x) + f_1(t) r(t,x)
    \end{split}
    \end{equation}
    where $r(t,x) = n^2(\phi(t,x+1/n)-2\phi(t,x)+\phi(t,x-1/n))-\Delta\phi(t,x)$, the vector $\vec{\phi}(t) = (\phi(t,k/n))_{k=0}^{n-1}$ satisfies the ordinary differential equation
    \begin{equation}
        \I \partial_t \vec{\phi}(t) = (f_1(t)H_1 + f_2(t)H_2)\vec{\phi}(t) + f_1(t)\vec{R}(t)
    \end{equation}
    where $\vec{R}(t) = (r(t,k/n))_{k=0}^{n-1}$. 
    Same as our previous notations, let $U(t,s)$ denote the evolution operator from time $s$ to $t$ of the dynamics \cref{eqn:td_hamsim} with Hamiltonian \cref{eqn:schrodinger_tdmass}. 
    By the variation of parameters formula (\cref{lem:VoP}), 
    \begin{equation}
        \vec{\phi}(t) = \vec{\psi}(t) + \int_0^t U(t,s) f_1(s) \vec{R}(s) ds, 
    \end{equation}
    and thus 
    \begin{equation}\label{eqn:schrodinger_tdmass_space_error_est}
        \|\vec{\phi}(t)-\vec{\psi}(t)\|_{\star} \leq  t \|f_1\|_{\infty}\sup_{s\in[0,t]}\|\vec{R}(s)\|_{\star}. 
    \end{equation}
    It remains to bound $\|\vec{R}(s)\|_{\star}$. 
    
    By the definition of $\vec{R}$ and the first part of this lemma, for any $s$, 
    \begin{equation}
        \begin{split}
            \|\vec{R}(s)\|_{\star}^2 &= \frac{1}{n} \sum_{k=0}^{n-1} |r(s,k/n)|^2 \\
            & \leq \frac{1}{n}\sum_{k=0}^{n-1} \left(\frac{1}{3n^2} \sup_{y\in[0,1]} \left|\frac{\partial^4}{\partial x^4}\phi(s,y)\right|\right)^2 \\
            & \leq \frac{1}{9n^4} \left( \sup_{y\in[0,1]} \left|\frac{\partial^4}{\partial x^4}\phi(s,y)\right|\right)^2. 
        \end{split}
    \end{equation}
    Plug this estimate back to \cref{eqn:schrodinger_tdmass_space_error_est}, we complete the proof. 
    
\end{proof}

\end{document}